\documentclass{article}
\usepackage{amsthm,amsmath,amssymb,verbatim,algorithmic,algorithm,tikz,mathdots,mathpartir,stmaryrd,fancyvrb,MnSymbol,RR,RRthemes}

\newtheorem{theorem}{Theorem}[section]
\newtheorem{lemma}[theorem]{Lemma}
\newtheorem{proposition}[theorem]{Proposition}

\newcommand{\node}{n}  % nodes in a tree
\newcommand{\tree}{T}   % tree structures
\newcommand{\Nodes}{N}   % set of nodes of  a tree structure
\newcommand{\brel}{R}   % binary relation in a tree
\newcommand{\tlabel}{L}  % label function in a tre
\newcommand{\path}[3]{#1 \stackrel{#2}{\longrightarrow} #3}  % a path from #1 to #2 throught #3
\newcommand{\natn}{k}  % natural number
\newcommand{\Modals}{M} % the set of modalities
\newcommand{\Props}{P} % the set of propositions
\newcommand{\trail}{\alpha}  % a trail in a tree
\newcommand{\Trails}{Tl}  %  set of trails
\newcommand{\emptytrail}{\epsilon}
\newcommand{\dual}[1]{\overline{#1}}
\newcommand{\eqdef}{\stackrel{\text{\tiny{def}}}{=}}

\newcommand{\syntaxdef}{\mathrel{::=}}

\newcommand{\syntaxtable}[1]{
  \def\entry##1[##2]##3[##4]{
    {##1} & \syntaxdef& \hspace{3cm} & \!\!\!\! \mbox{##2}
    \\    &     & {##3} & \mbox{##4} }
  \def\singleentry##1[]##2[##3]{
  {##1} & \syntaxdef& {##2} & \!\!\!\! \mbox{##3} }
  \def\oris##1[##2]{
    \\    & |   & {##1} & \mbox{##2} }
  \def\orisopt##1[##2]{
    \\ \left(& |   & {##1} & \mbox{##2} \right) }
  \begin{array}{rcll}
  #1
  \end{array}
  }
\newcommand{\smallsyntax}[1]{\[\syntaxtable{#1}\]}
\newcommand{\modals}{m}  % modal symbol (modality)
\newcommand{\Formulas}{\Phi}  % the set of formulas

\newcommand{\true}{\top}  % the true symbol
\newcommand{\var}{x}  % variables
\newcommand{\prop}{p} % propositions
\newcommand{\form}{\phi}  % formulas
\newcommand{\restrictedform}{\psi}  % restricted formulas
\newcommand{\modalf}[2]{\langle #1\rangle #2} % modal formulas  (modality, formula)
\newcommand{\ufixpf}[2]{\mu #1.#2}  % unary fixpoint formula (variable, formula)
\newcommand{\countf}[5]{\langle {#2} \rangle^{#1}_{#4#5}#3}

\newcommand{\valuation}{V}  % valuation relation among nodes and variables
\newcommand{\semf}[3]{[\![#1]\!]^{#2}_{#3}}  % semantics of a formula #1 w.r.t tree #2 and a valuation #3
\newcommand{\semfc}[3]{\semf{#1}{#2_c}{#3}}

\newcommand{\flcl}[1]{{FL}({#1})}   % fisher-ladner closure
\newcommand{\lean}[1]{{lean}({#1})}  % lean set
\newcommand{\fnode}[1]{n^{#1}}   % formula-node
\newcommand{\fNodes}[1]{N^{#1}}  % multiset of formula-nodes
\newcommand{\fbrel}[4]{R^{#1}(#2,#3)= #4}  % transition relation among formula-nodes through modalities (formula, node1, modality, node2)
\newcommand{\stree}{\Gamma}   % syntactic tree
\newcommand{\stroot}[1]{root(#1)}   %   the root of a syntactic tree
\newcommand{\stnodes}[1]{nodes(#1)}  % the nodes of a syntactic tree
\newcommand{\nav}[2]{nav((#1),#2)}  % (trail, formula)
\newcommand{\fishrel}[2]{ R^{fl}(#1,#2)}

\newcommand{\pathvar}{\rho} % path expression variable
\newcommand{\qvar}{q} % qualifier variable
\newcommand{\axvar}{a} % axis variable
\newcommand{\ch}{\text{child}}
\newcommand{\self}{\text{self}}
\newcommand{\pnt}{\text{parent}}
\newcommand{\desc}{\text{descendant}}
\newcommand{\descsf}{\text{desc$\!-\!$or$\!-\!$self}}
\newcommand{\anc}{\text{ancestor}}
\newcommand{\ancsf}{\text{anc$\!-\!$or$\!-\!$self}}
\newcommand{\fsib}{\text{foll$\!-\!$sibling}}
\newcommand{\psib}{\text{prec$\!-\!$sibling}}
\newcommand{\sibs}{\text{siblings}}
\newcommand{\foll}{\text{following}}
\newcommand{\prdn}{\text{preceding}}

\newcommand{\syntaxdefinition}{::=}

\newcommand{\XPath}{\text{XPath}}
\newcommand{\Axis}{\text{Axis}}
\newcommand{\PathExpr}{\text{PathExpr}}
\newcommand{\Step}{\text{Step}}
\newcommand{\NameTest}{\text{NameTest}}

\newcommand{\Qualifier}{\text{Qualifier}}
\newcommand{\QName}{\text{QName}}
\newcommand{\pathexcept}{~\text{except}~}
\newcommand{\pathunion}{~\text{union}~}
\newcommand{\pathintersect}{~\text{intersect}~}
\newcommand{\qualifnot}{\text{not}~}
\newcommand{\qualifand}{~\text{and}~}
\newcommand{\qualifor}{~\text{or}~}
\newcommand{\CountExpr}{\text{CountExpr}}
\newcommand{\Comparison}{\text{Comp}}
\newcommand{\qualifposition}{\text{position}()}
\newcommand{\qualifcount}[1]{\text{count}(#1)}

\newcommand{\precedence}{\ll}

\newcommand{\dom}{N}
\newcommand{\pathsem}[1]{\llbracket #1 \rrbracket}
\newcommand{\qualifsem}[1]{\llbracket #1 \rrbracket_{\text{Qualif}}}
\newcommand{\card}[1]{|#1|}

\newcommand{\fc}{\medtriangledown}%\downarrow}%\swarrow}%\text{first}}%
\newcommand{\ns}{\medtriangleright}%\text{next}}%
\newcommand{\invfc}{\medtriangleup}%\uparrow}\dual{\fc}}%
\newcommand{\invns}{\medtriangleleft}%\dual{\ns}}

\newcommand{\mucalcEFunc}{E^\rightarrow}
\newcommand{\mucalcPFunc}{P^\rightarrow}
\newcommand{\mucalcAFunc}{A^\rightarrow}
\newcommand{\mucalcE}[2]{\mucalcEFunc\llbracket{#1}\rrbracket_{#2}} 
\newcommand{\mucalcP}[2]{\mucalcPFunc\llbracket{#1}\rrbracket_{#2}}
\newcommand{\mucalcA}[2]{\mucalcAFunc\llbracket{#1}\rrbracket_{#2}}

\newcommand{\mucalcPexFunc}{P^\leftarrow}
\newcommand{\mucalcQexFunc}{Q^\leftarrow}
\newcommand{\mucalcAexFunc}{A^\leftarrow}
\newcommand{\mucalcPex}[2]{\mucalcPexFunc\llbracket{#1}\rrbracket_{#2}}
\newcommand{\mucalcQex}[2]{\mucalcQexFunc\llbracket{#1}\rrbracket_{#2}}
\newcommand{\mucalcAex}[2]{\mucalcAexFunc\llbracket{#1}\rrbracket_{#2}}

\newcommand{\nominal}{@n}
\newcommand{\et}{\wedge}
\newcommand{\ou}{\vee}

\newcommand{\someFCverifies}{\left<\fc\right>}
\newcommand{\someNSverifies}{\left<\ns\right>}

\newcommand{\someinvFCverifies}{\left<\invfc\right>}
\newcommand{\someinvNSverifies}{\left<\invns\right>}

\newcommand{\step}[2]{\text{{#1}::}{#2}}
\newcommand{\axis}[1]{\text{#1}}
\newcommand{\axisvar}{\emph{a}}

\newcommand{\qualif}[2]{{#1}\text{[$#2$]}}
\newcommand{\op}[1]{\mathbin{\text{\small{#1}}}}

\newcommand{\nodelabel}{\sigma}

\newcommand{\startatom}{\circledS}

\newcommand{\extracttrail}[1]{\text{trail}(#1)}
\newcommand{\target}[1]{\text{tail}(#1)}
\newcommand{\head}[1]{\text{head}(#1)}

\newcommand{\subst}[2]{^{#1} \!/\! _{#2}}
\newcommand{\psitree}[1]{$#1$tree}
\newcommand{\pathn}{\rho}
\newcommand{\nmax}{\mathtt{nmax}}
\newcommand{\imax}{\mathtt{max}}

\newcommand{\slf}[1]{\text{sf}(#1)}

\newcommand{\uc}{\text{ch}}

\newcommand{\induced}{\stackrel{.}{\in}}

\begin{document}

\RRNo{7251}

\RRtitle%{A Tree Logic with Graded Multidirectional Paths}
{On the Count of Trees}%}%: Decidability and Succinctness} %Reasoning with 
%{Regular Path Expressions with Counting: a Succinct Tree Logic with Counting Along Multidirectional Paths}

\RRetitle%{A Tree Logic with Graded Multidirectional Paths}
{On the Count of Trees}%}%: Decidability and Succinctness} %Reasoning with 
%{Regular Path Expressions with Counting: a Succinct Tree Logic with Counting Along Multidirectional Paths}

\RRauthor{Everardo B\'arcenas 
\and Pierre Genev\`es 
\and Nabil Laya\"ida 
\and Alan Schmitt 
}

\RRresume{Ce document introduit une logique d'arbre d\'ecidable en temps exponentielle et qui est capable d'exprimer des contraintes de cardinalit\'e sur chemins multidirectionnelle.}
\RRkeyword{Modal Logic, XML, XPath, Schema}
\RRmotcle{Logique Modal, XML, XPath, Schema}
\RRprojets{WAM et SARDES}
%\RRdomaine{2} % cas du domaine numero 1
\RRthemeProj{wam} % theme du projet Apics
%\RRdomaineProjBis{wam}
\RRdate{April 2010}
\RRversion{2}
\RRdater{August 2010}

\URRhoneAlpes

\RRabstract{
Regular tree grammars and regular path expressions constitute core constructs
widely used in programming languages and type systems. Nevertheless, there has
been little research so far on frameworks for reasoning about path expressions
where node cardinality constraints occur along a path in a tree. We present a
logic capable of expressing deep counting along paths which may include
arbitrary recursive forward and backward navigation. The counting extensions
can be seen as a generalization of graded modalities that count immediate
successor nodes. While the combination of graded modalities, nominals, and
inverse modalities yields undecidable logics over graphs, we show that these
features can be combined in a decidable tree logic whose main features can be
decided in exponential time. Our logic being closed under negation, it may be
used to decide typical problems on XPath queries such as satisfiability, type
checking with relation to regular types, containment, or equivalence.
}

\makeRR

\section{Introduction}

A fundamental peculiarity of XML is the description of regular properties. For example, in XML schema languages the content types of element definitions rely on regular expressions. In addition, selecting nodes in such constrained trees is also done by means of regular path expressions (\`a la XPath). In both cases, it is often interesting to be able to express conditions on the frequency of occurrences of nodes. 

Even if we consider simple strings, it is well known that some formal languages easily described in English may require voluminous regular expressions. For instance, as pointed out in \cite{klarlund-tacas95}, the language $L_{2a2b}$ of all strings over $\Sigma=\{a,b,c\}$ containing at least two occurrences of $a$ \emph{and} at least two occurrences of $b$ requires a large expression, such as:
\begin{align*}
&& \Sigma^*a\Sigma^*a\Sigma^*b\Sigma^*b\Sigma^* &&\cup&& \Sigma^*a\Sigma^*b\Sigma^*a\Sigma^*b\Sigma^* \\
%-----
&\cup& \Sigma^*a\Sigma^*b\Sigma^*b\Sigma^*a\Sigma^* 
&&\cup&& \Sigma^*b\Sigma^*b\Sigma^*a\Sigma^*a\Sigma^* \\
%-----
&\cup& \Sigma^*b\Sigma^*a\Sigma^*b\Sigma^*a\Sigma^* 
&&\cup&& \Sigma^*b\Sigma^*a\Sigma^*a\Sigma^*b\Sigma^*.
\end{align*}
If we add $\cap$ to the operators for forming regular expressions, then the language  $L_{2a2b}$  can be expressed more concisely as $(\Sigma^*a\Sigma^*a\Sigma^*) \cap (\Sigma^*b\Sigma^*b\Sigma^*)$. In logical terms, conjunction offers a dramatic reduction in expression size, which is crucial when the complexity of the decision procedure depends on formula size.

If we now consider a formalism equipped with the ability to describe numerical constraints on the frequency of occurrences, we get a second (exponential) reduction in size. For instance, the above expression can be formulated as
 $(\Sigma^*a\Sigma^*)^2 \cap (\Sigma^*b\Sigma^*)^2$.
We can even write  $(\Sigma^*a\Sigma^*)^{2^{n}} \cap (\Sigma^*b\Sigma^*)^{2^{n}}$ (for any natural $n$) instead of a (much) larger expression.

Different extensions of regular expressions with intersection, counting constraints, and interleaving have been considered over strings, and for describing content models of sibling nodes in XML type languages \cite{ghelli-icdt09,Gelade-siam08,Kilpelainen-ic07}.  The complexity of the inclusion problem over these different language extensions and their combinations typically ranges from polynomial time to exponential space (see \cite{Gelade-siam08} for a survey). The main distinction between these works and the  work presented here is that we focus on counting nodes located along deep and recursive paths in trees.

When considering regular \emph{tree} languages instead of regular \emph{string} languages,  succinct syntax such as the one presented above is even more useful, as branching results in a higher combinatorial complexity. 
In the case of trees, it is often useful to express cardinality constraints not only on the sequence of children nodes, but also in a particular region of a tree, such as a subtree. Suppose, for instance, that we want to define a tree language over $\Sigma$ where there is no more than 2 ``b'' nodes. This requires a quite large regular tree type expression such as:
\newcommand{\Xstart}{x_\text{root}}
\newcommand{\Xtwobmax}{x_{b\leq2}}
\newcommand{\Xnob}{x_{\neg b}}
\newcommand{\Xonebmax}{x_{b\leq1}}
\newcommand{\lb}{\texttt{[}}
\newcommand{\rb}{\texttt{]}}
\newcommand{\xtag}[2]{#1 \lb #2 \rb}
\newcommand{\Xsimple}{x}
$$\begin{array}{lcl}
\Xstart  \!\!\! &\rightarrow&\!\!\ \xtag{b}{\Xonebmax} \mid \xtag{c}{\Xtwobmax} \mid \xtag{a}{\Xtwobmax} \vspace{0.1cm} \\
\Xtwobmax\!\!\! &\rightarrow&\!\!\! \Xnob, \xtag{b}{\Xnob},\Xnob,\xtag{b}{\Xnob},\Xnob \mid \Xnob, \xtag{b}{\Xonebmax},\Xnob  \vspace{0.1cm}\\
&& \mid \Xnob, \xtag{a}{\Xtwobmax}, \Xnob \mid \Xnob, \xtag{c}{\Xtwobmax}, \Xnob  \mid \Xonebmax \vspace{0.1cm}\\
\Xonebmax \!\!\! &\rightarrow&\!\!\ \Xnob \mid  \Xnob, \xtag{b}{\Xnob}, \Xnob \mid \xtag{a}{\Xonebmax} \mid \xtag{c}{\Xonebmax}  \vspace{0.1cm}\\
\Xnob \!\!\! &\rightarrow&\!\!\ (\xtag{a}{\Xnob} \mid \xtag{c}{\Xnob})^*
\end{array}$$
where $\Xstart$ is the starting non-terminal; $\Xnob, \Xonebmax, \Xtwobmax$ are non-terminals; the notation $\xtag{a}{\Xnob}$ describes a subtree whose root is labeled $a$ and in which there is no $b$ node; and ``,'' is concatenation.

More generally, the widely adopted notations for regular tree grammars produce very verbose definitions for properties involving cardinality constraints on the nesting of elements\footnote{This is typically the reason why the standard DTD for XHTML does not syntactically prevent the nesting of anchors, whereas this nesting is actually prohibited in the XHTML standard.}.

The problem with regular tree (and even string) grammars is that one is forced to fully expand all the patterns of interest using concatenation, union, and Kleene star. Instead, it is often tempting to rely on another kind of (formal) notation that just describes a simple pattern and additional constraints on it, which are intuitive and compact with respect to size. For instance, one could imagine denoting the previous example as follows, where the additional constraint is described using XPath notation:
$$\left(\Xsimple  \!\rightarrow\!\! (\xtag{a}{\Xsimple} \mid \xtag{b}{\Xsimple} \mid \xtag{c}{\Xsimple})^*\right) \vspace{0.1cm}  ~\et~ \text{count(/descendant-or-self::}b) \leq 2$$

Although this kind of counting operators does not increase the expressive power of regular tree grammars, it can have a drastic impact on succinctness, thus making reasoning over these languages harder (as noticed in \cite{DBLP:conf/mfcs/Gelade08} in the case of strings). Indeed, reasoning on this kind of extensions without relying on their expansion (in order to avoid syntactic blow-ups) is often tricky \cite{DBLP:conf/mfcs/GeladeGM09}. Determining satisfiability, containment, and equivalence over these classes of extended regular expressions typically requires involved algorithms with higher complexity \cite{DBLP:conf/focs/MeyerS72} compared to ordinary regular expressions.

In the present paper, we propose a succinct logical notation, equipped with a satisfiability checking algorithm, for describing many sorts of cardinality constraints on the frequency of occurrence of nodes in regular tree types. 
Regular tree types encompass most of XML  types  (DTDs, XML Schemas, RelaxNGs) used in practice today.

XPath is the standard query language for XML documents, and it is an important part of other XML technologies such as XSLT and XQuery.
XPath expressions are regular path expressions interpreted as sets of nodes selected from a given context node. 
One of the reasons why XPath is popular for web programming resides in its ability to express multidirectional navigation. Indeed, XPath expressions may use recursive navigation, to access descendant nodes, and also backward navigation, to reach previous siblings or ancestor nodes. 
Expressing cardinality restrictions on nodes accessible by recursive multidirectional paths may introduce an extra-exponential cost \cite{DBLP:conf/doceng/GenevesR05,Balder09}, or may even lead to undecidable formalisms \cite{Balder09,DBLP:conf/cade/DemriL06}.
We present in this paper a decidable framework capable of succinctly expressing cardinality constraints along deep multidirectional paths.

A major application of this logical framework is the decision of problems found in the static analysis of programming languages manipulating XML data. For instance, since the logic is closed under negation, it can be used to solve subtyping problems such as XPath containment in the presence of tree constraints. Checking that a query $q$ is contained in a query $p$ with this logical approach amounts to verifying the validity of $q\Rightarrow p$, or equivalently, the unsatisfiability of $q \wedge \neg p$.

\paragraph{Contributions}
We extend a tree logic with a succinct notation for counting operators. These operators allow arbitrarily deep and recursive counting constraints. We present a sound and complete algorithm for checking satisfiability of logical formulas. We show that its complexity is exponential in the size of the succinct form. 

\paragraph{Outline}
We introduce the logic in Section~\ref{sec:logic}.   Section~\ref{sec:application} shows how the logic can be applied in the XML setting, in particular for the static analysis of XPath expressions and of common schemas containing constraints on the frequency of occurrence of nodes. The decision procedure and the proofs of soundness, completeness, and complexity are presented in Section~\ref{sec:algo}.
Finally, we review related work in Section~\ref{sec:relatedwork} before concluding in Section~\ref{sec:conclusion}.

\section{Counting Tree Logic} \label{sec:logic}
We introduce our syntax for trees, define a notion of trails in trees, then present the syntax and semantics of logical formulas.

%%%%%%%%%%%%%%%%%%%%%%%%%%%%%%%%%%%%%%%%%%%%%%%%%%%%%%%%%%%%
\subsection{Trees}\label{subsec:trees}
 %Trails among two nodes in a tree are defined in order to introduce a counting operator able to perform a deep-like counting.
%Models of the logic are trees.\as{Do we need to say this now? We repeat that later, when defining the interpretation of formulas.} 
%
We consider finite trees which are node-labeled and sibling-ordered. Since there is a well-known bijective encoding between $n\!-\!$ary and binary trees, we focus on binary trees without loss of generality.
Specifically, we use the encoding represented in Figure~\ref{fig:depthlevels}, 
where the binary representation preserves the first child of a node and append sibling nodes as second successors.
\begin{figure}
\begin{center}
\includegraphics[keepaspectratio,width=7cm]{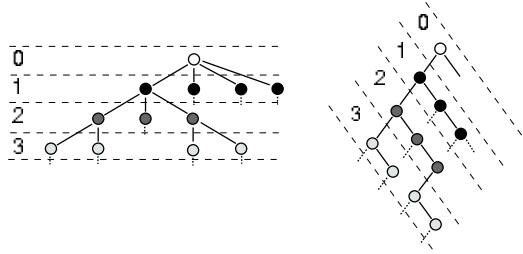}
\end{center}
\caption{$n\!-\!$ary to binary trees}
\label{fig:depthlevels}
\end{figure}

The structure of a tree is built upon modalities ``$\fc$'' and  ``$\ns$''. Modality  ``$\fc$'' labels the edge between a node and its first child. Modality  ``$\ns$'' labels the edge between a node and its next sibling.
Converse modalities  ``$\invfc$'' and  ``$\invns$'' respectively label the same edges in the reverse direction.

We define a Kripke semantics for our tree logic, similar to the one of modal logics \cite{685998}. We write   $\Modals=\{\fc,\ns,\invfc,\invns\}$ for the set of modalities. For $m\in \Modals$ we denote by $\dual{m}$ the corresponding inverse modality ($\dual{\fc}=\invfc, \dual{\ns}=\invns, \dual{\invfc}=\fc, \dual{\invns}=\ns$). We also consider a countable alphabet $\Props$ of {\em propositions} representing names of nodes. A node is always labeled with exactly one proposition.

A tree is defined as a tuple $(\Nodes,\brel,\tlabel)$, where $\Nodes$ is a finite set of nodes; $\brel$ is a partial mapping from $\Nodes\times\Modals$ to $\Nodes$ that defines a tree structure;\footnote{For all $n,n' \in \Nodes,m \in \Modals$, $\brel(n,m) = n' \iff \brel(n', \dual{m}) = n$; for all $n \in \Nodes$ except one (the \emph{root}), exactly one of $\brel(n,\invfc)$ or $\brel(n,\invns)$ is defined; for the root, neither $\brel(n,\invfc)$ nor $\brel(n,\invns)$ is defined.} and $\tlabel$ is a labeling function from $\Nodes$ to $\Props$.

\subsection{Trails}\label{subsec:trails}

Trails are defined as regular expressions formed by modalities, as follows: %in Figure \ref{trailsyn}.
%\begin{comment} %%%%%%%%%%%%%%%%%%%%%%%%%%%%%%%%%%%%%%%%%%%%%%%%%%%%%%%%%
%%\begin{figure}
%\begin{align*}
%& \Trails \ni \trail &&:=&& &&\text{trail}\\
%& && && \emptytrail &&\text{empty}\\
%& &&|&& \modals     &&\text{modality}\\
%& &&|&& \trail_1,\trail_2 && \text{concatenation}\\
%& &&|&& \trail_1|\trail_2 && \text{disjunction}\\
%& &&|&& \trail^\star      && \text{recursion}
%\end{align*}
%\end{comment}%%%%%%%%%%%%%%%%%%%%%%%%%%%%%%%%%%%%%%%%%%%%%%%%%%%%%%%%%%
\begin{align*}
\trail &::= \trail_0 \mid \trail_0^\star \mid \trail_0^\star, \trail \\
\trail_0 &::= \modals \mid \trail_0,\trail_0 \mid \trail_0\shortmid\trail_0
\end{align*}
%\caption{Syntax of Trails.} \label{trailsyn}
%\end{figure}
%Two interpretations of trails are considered: a syntactic one and a semantic one.
We restrict trails to sequences of repeated subtrails (which themselves contain no repetition) followed by a subtrail (with no repetition). Since we do not consider infinite paths, we also disallow trails where both a subtrail and its converse occurs under the scope of the recursion operator, thus ensuring cycle-freeness (see Section \ref{cyclefreeness}). These restrictions on trails allow us to prove the completeness of our approach while retaining the ability to express many counting formulas, such as the ones of XPath.

Trails are interpreted as sets of \emph{paths}. A path, written $\pathn$, is a sequence of modalities that belongs to the regular language denoted by the trail, written $\pathn \in \trail$.

In a given tree, we say that there is a {\em trail $\trail$ from the node $\node_0$ to the node $\node_\natn$}, written $\path{n_0}{\trail}{n_\natn}$,
if and only if there is a sequence of nodes $\node_0, \ldots, n_k$ and a path $\pathn = \modals_1, \ldots, \modals_k$ such that $\pathn \in \trail$, and $\brel(\node_j,\modals_{j+1})=\node_{j+1}$ for every $j=0,\ldots,\natn-1$.  

\subsection{Syntax of Logical Formulas}\label{subsec:syntax}

The syntax of logical formulas is given in Figure \ref{normalformsyn}, where $m \in \Modals$ and $k \in \mathbb{N}$. Formulas written $\form$ may contain counting subformulas, whereas formulas written $\restrictedform$ cannot. We thus disallow counting under counting or under fixpoints. We also restrict formulas to \emph{cycle-free} formulas, as detailed in Section~\ref{cyclefreeness}. The syntax is shown in negation normal form. The negation of any closed formula (i.e., with no free variable) built using the syntax of Figure \ref{normalformsyn} may be transformed into negation normal form using the usual De Morgan rules together with rules given in Figure~\ref{normalneg}. 
%The fact that the semantic interpretation is preserved even though the smallest fixpoint does not become a greatest fixpoint is a consequence of Lemma~\ref{lem:finite_unfolding}. Hence, 
When we write $\neg\form$, we mean its negated normal form.
\begin{figure}
  
  \smallsyntax{
  \Formulas \ni \entry       \form    [formula]
    \top  \quad | \quad  \neg\top           [true, false]
  \oris \prop ~\quad | \quad \neg \prop             [atomic prop (negated)]
  \oris        x [recursion variable]
  \oris        \phi \ou \phi [disjunction]
  \oris        \phi \et \phi [conjunction]
  \oris  \modalf{\modals}\phi  \quad | \quad  \neg \modalf{\modals}\true [modality (negated)]

  \oris \countf{}{\trail}{\restrictedform}{\leq}{\natn}  \quad | \quad \countf{}{\trail}{\restrictedform}{>}{\natn} [counting]

  \oris        \ufixpf{\var}{\restrictedform} [fixpoint operator] \\
  \entry       \restrictedform    []
    \top  \quad | \quad  \neg\top           []
  \oris \prop ~\quad | \quad \neg \prop             []
  \oris        x []
  \oris        \restrictedform \ou \restrictedform []
  \oris        \restrictedform \et \restrictedform []
  \oris  \modalf{\modals}\restrictedform  \quad | \quad  \neg \modalf{\modals}\true []
  \oris        \ufixpf{\var}{\restrictedform} []
  } 
\caption{Syntax of Formulas (in Normal Form).} \label{normalformsyn}
\end{figure}

\begin{figure}
\begin{align*}
\neg \modalf{\modals}{\form}  &\equiv  \neg \modalf{\modals}{\true} \vee \modalf{\modals}{\neg \form} 
& \neg \ufixpf{\var}{\restrictedform}  &\equiv \ufixpf{\var}{\neg \restrictedform\{\subst{\var}{\neg \var}\}}\\
 \neg \countf{}{\trail}{\restrictedform}{\leq}{\natn} &\equiv \countf{}{\trail}{\restrictedform}{>}{\natn} &  \neg \countf{}{\trail}{\restrictedform}{>}{\natn} &\equiv \countf{}{\trail}{\restrictedform}{\leq}{\natn}
\end{align*}
\caption{Reduction to Negation Normal Form.}\label{normalneg}
\end{figure}

Defining an {\em equality} operator for counting formulas is straightforward using the other counting operators.
\begin{align*}
\countf{}{\trail}{\restrictedform}{=}{\natn} &\equiv \countf{}{\trail}{\restrictedform}{>}{(\natn-1)} \wedge \countf{}{\trail}{\restrictedform}{\leq}{\natn} &&\text{if $k>0$}\\
\countf{}{\trail}{\restrictedform}{=}{0} &\equiv  \countf{}{\trail}{\restrictedform}{\leq}{0}
\end{align*}

\subsection{Semantics of Logical Formulas}\label{subsec:semantics}
A formula is interpreted as a set of nodes in a tree. A model of a formula is a tree such that the formula denotes
a non-empty set of nodes in this tree. 
A counting formula $\countf{}{\trail}{\restrictedform}{>}{\natn}$ satisfied at a given node $n$ means that there are at least $\natn + 1$  nodes satisfying $\restrictedform$ that can be reached from $n$ through the trail $\trail$. 
A counting formula $\countf{}{\trail}{\restrictedform}{>}{\natn}$ is thus interpreted as the set of nodes such that, for each of them, the previously described condition holds. For example, the formula $\prop_1\wedge \modalf{\fc}{\countf{}{\ns^*}{\prop_2}{>}{5}}$, denotes $\prop_1$ nodes
with strictly more than $5$ children nodes named $\prop_2$.

In order to present the formal semantics of formulas, we introduce valuations, written $\valuation$, which relate variables to sets of nodes.
We write $\valuation[\subst{\Nodes^\prime}{x}]$, where $\Nodes^\prime$ is a subset of the nodes, for the valuation defined as $\valuation[\subst{\Nodes^\prime}{x}](y) = \valuation(y)$ if $x \neq y$, and $\valuation[\subst{\Nodes^\prime}{x}](x) = \Nodes^\prime$.
%
%NDP: clean the definition of valuations: suggestion: A valuation $\valuation$ is an assignment of subsets of $\Nodes$ to variables....  
%
Given a tree $\tree=(\Nodes,\brel, \tlabel)$ and a valuation $\valuation$, the formal semantics of formulas is given in Figure~\ref{formsem}.

%Fig. 4: It deserves mentioning that F(Y)=[[psi]]^T_V[Y/x] is monotone or more generally that the denotation of mu x.psi is a fixed point.

%Usually where negation and mu appear together, there is some syntactic restriction (such as positivity) that implies monotonicity and ensures the existence of least fixpoints.  Here only syntactically monotone formulas are allowed, but then the claim is made that formulas are closed under negation, and negation is introduced as a meta-construct, but used as if it were part of the language.  This is rather unusual.  It would be better to have negation explicitly in the language if it is going to be used that way.

%A variable
%[[psi]]^T_V[Y/x]$

Note that the function $f :Y \rightarrow \semf{\psi}{\tree}{\valuation[\subst{Y}{x}]}$ is monotone, and the denotation of $\ufixpf{\var}{\restrictedform}$ is a fixed point \cite{tarski1955}.

\begin{figure}
\begin{align*}
& \semf{\true}{\tree}{\valuation} &&=&& \Nodes \\
& \semf{\neg \true}{\tree}{\valuation} &&=&& \emptyset \\
& \semf{\prop}{\tree}{\valuation} &&=&& \{\node \mid \tlabel(\node)=\prop\} \\
& \semf{\neg \prop}{\tree}{\valuation} &&=&& \{\node \mid \tlabel(\node)\neq \prop\} \\
& \semf{\var}{\tree}{\valuation} &&=&& \valuation(\var)\\
& \semf{\form_1\vee \form_2}{\tree}{\valuation} &&=&& \semf{\form_1}{\tree}{\valuation} \cup \semf{\form_2}{\tree}{\valuation}\\
& \semf{\form_1\wedge \form_2}{\tree}{\valuation} &&=&& \semf{\form_1}{\tree}{\valuation} \cap \semf{\form_2}{\tree}{\valuation}\\
& \semf{\modalf{\modals}{\form}}{\tree}{\valuation} &&=&& \{\node \mid \brel(\node,\modals) \in \semf{\form}{\tree}{\valuation}\}\\
& \semf{\neg \modalf{\modals}{\true}}{\tree}{\valuation} &&=&& \{\node \mid \brel(\node,\modals) \text{ undefined}\} \\
& \semf{\countf{}{\alpha}{\restrictedform}{\leq}{\natn}}{\tree}{\valuation} &&=&&
    \{\node \mid~~ \card{\{\node^\prime\in \semf{\restrictedform}{\tree}{\valuation} \mid \path{\node}{\alpha}{\node^\prime}\}}\leq  \natn \}\\
& \semf{\countf{}{\alpha}{\restrictedform}{>}{\natn}}{\tree}{\valuation} &&=&&
    \{\node \mid~~  \card{\{\node^\prime\in \semf{\restrictedform}{\tree}{\valuation} \mid \path{\node}{\alpha}{\node^\prime}\}}>  \natn \}\\
& \semf{\ufixpf{\var}{\restrictedform}}{\tree}{\valuation} &&=&& \bigcap\{\Nodes^\prime \mid \semf{\restrictedform}{\tree}{\valuation[\subst{\Nodes^\prime}{x}]}\subseteq \Nodes^\prime\}
\end{align*}
\caption{Semantics of Formulas.} \label{formsem}
\end{figure}

Intuitively, propositions denote the nodes where they occur;
negation is interpreted as set complement; disjunction and conjunction are respectively set union and intersection;
the least fixpoint operator performs finite recursive navigation; 
and the counting operator denotes nodes such that the ones accessible from this node through a trail fulfill a cardinality restriction. A logical formula is said to be {\em satisfiable} iff it has a model, i.e., there exists a tree for which the semantics of the formula is not empty.

% Given a tree $\tree$, a formula is said to be {\em satisfied} by 
%$\tree$ when its interpretation over $\tree$ is not empty.
%If there is a tree satisfying a formula, we said the formula is {\em satisfiable}.

%there exists a tree for which it does not denote the empty set.

%When two formulas denote the same set of nodes for every tree, we say such formulas are {\em equivalent}.\as{I guess equivalent should be ``for every tree''.}

\subsection{Cycle-Freeness} \label{cyclefreeness}
Formal definition of cycle-freeness can be found in \cite{geneves-pldi07}.
Intuitively, in a cycle-free formula, fixpoint variables must occur under a modality but cannot occur in the scope of both a modality and its converse. 
For instance, the formula $\ufixpf{\var}{\modalf{\fc}{\var}\vee\modalf{\invfc}{\var}}$ is not cycle-free.
In a cycle-free formula, the number of modality cycles (of the form $m \dual{m}$) is bound independently of the number of times fixpoints are unfolded (i.e., by replacing a fixpoint variable with the fixpoint itself). A fundamental consequence of the restriction to cycle-free formulas is that, when considering only finite trees, the interpretations of the greatest and smallest fixpoints coincide. This greatly simplifies the logic.

Here, we also restrict our approach to cycle-free formulas. We thus need to extend this notion to the counting operators, and more precisely to the trails that occur in them.
Cycle-free trails are trails where both a subtrail and its converse do not occur under the scope of the recursion operator. We thus restrict the formulas under consideration to cycle-free formulas whose counting operators contain cycle-free trails.

\begin{lemma}\label{lem:finite_unfolding}
  Let $\form$ be a cycle-free formula, and $\tree$ be a tree for which $\semf{\form}{\tree}{\emptyset} \neq \emptyset$. Then there is a finite unfolding $\form^\prime$ of the fixpoints of $\form$ such that $\semf{\form^{\prime}\{\subst{\neg\true}{\ufixpf{\var}{\restrictedform}}\}}{\tree}{\emptyset} = \semf{\form}{\tree}{\emptyset}$.
\end{lemma}

\begin{proof}
  As cycle-free counting formulas may be translated into (exponentially larger) cycle-free non-counting formulas, the proof is identical to the one in \cite{geneves-pldi07}.
\end{proof}

As a consequence, our logic is closed under negation even without greatest fixpoints.

\subsection{Global Counting Formulas and Nominals}\label{subsec:globalformula}

To conclude this section, we turn to an illustration of the expressive power of our logic. An interesting consequence of the inclusion of backward axes in trails is the ability to reach every node in the tree from any node of the tree, using the trail $(\invfc|\invns)^\star,(\fc|\ns)^\star$.\footnote{Note that this trail is cycle-free.} We can thus select some nodes depending on some global counting property. Consider the following formula, where $\#$ stands for one of the comparison operators $\leq$, $>$, or $=$.
\[\countf{}{(\invfc|\invns)^\star,(\fc|\ns)^\star}{\form_1}{\#}{\natn}\]
Intuitively, this formula counts how many nodes in the whole tree satisfy $\form_1$. For each node of the tree, it selects it if and only if the count is compatible with the comparison considered. The interpretation of this formula is thus either every node of the tree, or none.
It is then easy to restrict the selected nodes to some that satisfy another formula $\form_2$, using intersection.
\[(\countf{}{(\invfc|\invns)^\star,(\fc|\ns)^\star}{\form_1}{\#}{\natn}) \wedge \form_2\]
This formula select every node satisfying $\form_2$ if and only if there are $\#\natn$ nodes satisfying $\form_1$, which we write as follows.
\[ \form_1\#\natn \implies \form_2 \]
We can now express existential properties, such as ``select every node satisfying $\form_2$ if there exists a node satisfying $\form_1$''.
\[ \form_1 > 0 \implies \form_2 \]
We can also express universal properties, such as ``select every node satisfying $\form_2$ if every node satisfies $\form_1$''.
\[ (\neg\form_1) \leq 0 \implies \form_2 \]

Another way to interpret global counting formulas is as a generalization of the so-called nominals in the modal logics community~\cite{DBLP:conf/cade/SattlerV01}. Nominals are special propositions whose interpretation is a singleton (they occur exactly once in the model). They come for free with the logic. A nominal, denoted ``$\nominal$'', corresponds to the following global counting formula:
\[[\countf{}{(\invfc|\invns)^\star,(\fc|\ns)^\star}{n}{=}{1}] \wedge n \] where $n$ is a new fresh atomic proposition.

%We can also express a navigation to everywhere in the tree by the following fixpoint formula.
%\[EW(\form)= \ufixpf{\var_1}{(\ufixpf{\var_2}{\form \vee \modalf{\fc}{\var_2} \vee \modalf{\ns}{\var_2}}) \vee \modalf{\invfc}{\var_1} \vee \modalf{\invns}{\var_2}}\]
One may need for nominals to occur in the scope of counting formulas. As we disallow counting under counting, we propose the following alternative encoding of nominals in these cases: 
\begin{align*}
\nominal \equiv n \wedge\neg [& \desc(n) \vee \anc(n) \vee \\
                       %     &  \descsf(\sibs(n)) \vee 
							& \ancsf(\sibs(\descsf(n))) ], 
\end{align*}
where:
\begin{align*}
\desc(\restrictedform)&= \modalf{\fc}{\ufixpf{\var}{\restrictedform \vee \modalf{\fc}{\var} \vee \modalf{\ns}{\var}}};\\
 \fsib(\restrictedform) &= \ufixpf{\var}{\modalf{\ns}{\restrictedform} \vee \modalf{\ns}{\var}};\\
 \psib(\restrictedform) &= \ufixpf{\var}{\modalf{\invns}{\restrictedform} \vee \modalf{\invns}{\var}};\\
 \descsf(\restrictedform) &= \ufixpf{x}{\restrictedform \vee \modalf{\fc}{\ufixpf{y}{x \vee \modalf{\ns}{y}} }};\\
 \anc(\restrictedform) &= \ufixpf{\var}{\modalf{\invfc}{(\restrictedform \vee \var)} \vee \modalf{\invns}{\var}};\\ 
 \ancsf(\restrictedform)   &=  \ufixpf{x}{\restrictedform \vee \ufixpf{y}{\modalf{\invfc}{(y \vee x)} \vee \modalf{\invns}{y}} };\\
 \sibs(\restrictedform) &= \fsib(\restrictedform) \vee \psib(\restrictedform).
\end{align*}

%%%%%%%%%%%%%%%%%%%%%%%%%%%%%%%%%%%%%%%%%%%%%%%%%%%%%%%%%%%%%%%%%%%%%%%%%%%%%%%%%%%%%%%%%%%%%%%%%%%%%%%%%%%%%%

\section{Application to XML Trees} \label{sec:application} %Regular Paths and Counting

\subsection{XPath Expressions}\label{typespaths} \label{sec:xpath}

XPath \cite{xpath} was introduced as part of the W3C XSLT transformation language to have a non-XML format for selecting nodes and computing values from an XML
document (see \cite{geneves-pldi07} for a formal presentation of XPath). Since then, XPath has become part of
several other standards, in particular it forms the ``navigation subset'' of the XQuery language.

In their simplest form XPath expressions look like ``directory navigation paths''.  For example, the XPath
\begin{verbatim}
  /company/personnel/employee
\end{verbatim}
navigates from the root of a document through the top-level ``company'' node to its
``personnel'' child nodes and on to its ``employee'' child nodes.  The result of
the evaluation of the entire expression is the set of all the ``employee''
nodes that can be reached in this manner.  At each step in the navigation, the selected nodes for that step can be
filtered with a predicate test. Of special interest to us are the predicates that count nodes or that test the position of the selected node in the previous step's selection. For example, if we ask for
\begin{verbatim}
  /company/personnel/employee[position()=2]
\end{verbatim}
then the result is \emph{all} employee nodes that are the \emph{second} employee node (in document order) among
the employee child nodes of each personnel node selected by the previous step.

XPath also makes it possible to combine the capability of searching along
``axes'' other than the shown ``children of'' with counting constraints. 
%The essential XPath axes are illustrated on~Figure~\ref{fig:xpath-axes}.
% \begin{figure}
%\begin{center}
%\includegraphics[keepaspectratio=true, width=8cm]{simple-partition-no-shadow.pdf}
%\end{center}
%\caption{XPath axes: partition of tree nodes.}\label{fig:xpath-axes}
%\end{figure}
For example, if we ask for
\begin{verbatim}
/company[count(descendant::employee)<=300]/name
\end{verbatim}
then the result consists of the company names with less than 300 employees in total (the axis ``descendant'' is the transitive closure of the default -- and often omitted -- axis ``child''). 
%This illustrates the expressive power of the XPath language when extended with counting features.

The syntax and semantics of Core XPath expressions are respectively given on Figure~\ref{fig:xpath-syntax} and Figure~\ref{fig:xpath-semantics}.
An XPath expression is interpreted as a relation between nodes. The considered XPath fragment allows absolute and relative paths, path union, intersection, composition, as well as node tests and qualifiers with counting operators, conjunction, disjunction, negation, and path navigation. Furthermore, it supports all XPath axes allowing multidirectional navigation.

\begin{figure}
\begin{align*}
   \Axis      \syntaxdefinition &  \text{self} \mid \text{child} \mid \text{parent} \mid \text{descendant} \mid \text{ancestor} \mid\\ 
   &   \text{following-sibling}    \mid \text{preceding-sibling} \mid \\ 
   & \text{following} \mid \text{preceding} \\
  % ----
  \NameTest  \syntaxdefinition &  \QName \mid * \\
   % -----
  \Step  \syntaxdefinition & \Axis\text{::}\NameTest \\
   % -----
   \PathExpr  \syntaxdefinition &  \PathExpr/\PathExpr
\mid \PathExpr[\Qualifier] \mid \Step \\
 \Qualifier  \syntaxdefinition & \PathExpr \mid \CountExpr \mid\qualifnot \Qualifier \mid \\
  & \Qualifier \qualifand \Qualifier  \mid \Qualifier \qualifor \Qualifier \mid \nominal \\
  \CountExpr \syntaxdefinition & \qualifcount{\PathExpr'}~\Comparison~k \\ 
% no counting under counting:
 \PathExpr'  \syntaxdefinition &  \PathExpr'/\PathExpr'
\mid \PathExpr'[\Qualifier'] \mid \Step \\
\Qualifier'  \syntaxdefinition & \PathExpr' \mid\qualifnot \Qualifier' \mid \Qualifier' \qualifand \Qualifier' \\ 
& \mid \Qualifier' \qualifor \Qualifier' \mid \nominal \\
 \Comparison  \syntaxdefinition & \leq  \mid > \mid \geq   \mid < \mid =\\
  \XPath \syntaxdefinition & \PathExpr \mid /\PathExpr  \mid \XPath \pathunion \PathExpr \mid\\
   & \XPath \pathintersect \PathExpr  \mid \XPath \pathexcept \PathExpr 
   \end{align*}
   \caption{Syntax of Core XPath Expressions.}\label{fig:xpath-syntax}
\end{figure}

\begin{figure}
\begin{align*}
   \pathsem{ \Axis\text{::}\NameTest }      =& \{ (x,y) \in \dom^2 \mid x(\Axis)y \text{ and }  \\ 
                                                            & y \text{ satisfies } \NameTest  \} \\
   \pathsem{ /\PathExpr }    =&  \{ (r,y) \in \pathsem{\PathExpr } \mid \\ & r \text{ is the root}   \} \\
  % ----
   \pathsem{ P_1/P_2}    =& \pathsem{P_1} \circ \pathsem{P_2}\\
    % ----
   \pathsem{P_1 \pathunion P_2}    =&  \pathsem{P_1} \cup \pathsem{P_2}\\
     % ----
   \pathsem{ P_1 \pathintersect P_2 }    =& \pathsem{P_1} \cap \pathsem{P_2}\\
     % ----
   \pathsem{ P_1 \pathexcept P_2 }    =& \pathsem{P_1} \setminus \pathsem{P_2}\\
     % ----
   \pathsem{ \PathExpr[\Qualifier] }    = & \{ (x,y) \in \pathsem{\PathExpr} \mid \\ & y \in \qualifsem{\Qualifier} \}\\ \\
     % ----
   \qualifsem{  \PathExpr  }    =& \{ x \mid  \exists y. (x,y) \in \pathsem{\PathExpr}  \} \\
      %---
   \qualifsem{\qualifcount{\PathExpr}~\Comparison~k } =& \{ x \in \dom \mid \\& \card{\left\{y \mid (x,y) \in \pathsem{\PathExpr} \right\}} \\ & \text{satisfies }  \Comparison~k \} \\
     % ----
   \qualifsem{  \qualifnot Q  }   =& \dom \setminus \qualifsem{Q}\\
     % ----
   \qualifsem{Q_1 \qualifand Q_2}    =& \qualifsem{Q_1} \cap \qualifsem{Q_1} \\
     % ----
   \qualifsem{ Q_1 \qualifor Q_2 }   =&  \qualifsem{Q_2} \cup \qualifsem{Q_2}
   \end{align*}
   \caption{Semantics of Core XPath Expressions}\label{fig:xpath-semantics}
 \end{figure}

It was already observed in \cite{DBLP:conf/doceng/GenevesR05,Balder09} that using positional information in paths reduces to counting (at the cost of an exponential blow-up).
For example, the expression
\begin{verbatim}
child::a[position()=5]
\end{verbatim}
first selects the ``\texttt{a}'' nodes occurring as children of the current context node, and then keeps those occurring at the $5$th position. This expression can be rewritten into the semantically equivalent expression:
\begin{verbatim}
child::a[count(preceding-sibling::a)=4]
\end{verbatim}
which constraints the number of preceding siblings named ``\texttt{a}'' to $4$, so that the  qualifier becomes true only for the $5$th child ``\texttt{a}''. A general translation of positional information in terms of counting operators \cite{DBLP:conf/doceng/GenevesR05,Balder09} is summarized on Figure~\ref{fig:positional-info}, where $\precedence$ denotes the document order (depth-first left-to-right) relation in a tree.  Note that translated path expressions can in turn be expressed into the core XPath fragment of Figure~\ref{fig:xpath-syntax} (at the cost of another exponential blow-up). Indeed, expressions like $\PathExpr/(\PathExpr_2 \pathexcept \PathExpr_3)/\PathExpr_4$ must be rewritten into expressions where binary connectives for paths occur only at top level, as in:
\begin{align*}
& \PathExpr/\PathExpr_2/\PathExpr_4 \pathexcept \\ & \PathExpr/\PathExpr_3/\PathExpr_4
\end{align*}

\begin{figure}
\begin{align*}
\PathExpr[\qualifposition=1] \equiv &  \PathExpr \pathexcept (\PathExpr/\precedence) \\
\PathExpr[\qualifposition=k+1] \equiv &  (\PathExpr \pathintersect \\ & (\PathExpr[k]/\!\precedence))[\qualifposition\!=\!1] \\
\precedence \equiv&  (\text{descendant::*})  \pathunion (\text{a-o-s::*}/ \\ &  \text{following-sibling::*}/\text{d-or-s::*}) \\
\text{a-or-s::*}   \equiv & \text{ancestor::*} \pathunion \text{self::*}\\
\text{d-or-s::*} \equiv & \text{descendant::*} \pathunion \text{self::*}
\end{align*}
   \caption{Positional Information as Syntactic Sugars  \cite{DBLP:conf/doceng/GenevesR05,Balder09}}\label{fig:positional-info}
\end{figure}

We focus on Core XPath expressions involving the counting operator  (see Figure~\ref{fig:xpath-syntax}). The XPath fragment without the counting operator (the navigational fragment) was already linearly translated into $\mu$-calculus in  \cite{geneves-pldi07}. The contributions presented in this paper allow to equip this navigational fragment with counting features such as the ones formulated above. Logical formulas capture the aforementioned XPath counting constraints.
For example, consider the following XPath expression:
\begin{verbatim}
child::a[count(descendant::b[parent::c])>5]
\end{verbatim}
This expression selects the children nodes named ``\texttt{a}'' provided they have more than $5$
descendants which (1) are named ``\texttt{b}'' and (2) whose parent is named ``\texttt{c}''. The logical formula denoting the set of children nodes
named ``\texttt{a}'' is:
\[\psi = a\wedge \modalf{\invns^*,\invfc}\top  \]
The  logical translation of the above XPath expression is:
%\[\psi\wedge\countf{}{\fc,(\fc|\ns)^\star}{(b \wedge \mu y.\modalf{\invfc}c \vee \modalf{\invns}y)}{>}{5}\]
\[\psi\wedge\modalf{\fc}{\countf{}{(\fc|\ns)^\star}{(b \wedge \ufixpf{\var}{\modalf{\invfc}{c}\vee \modalf{\invns}{\var}} )}{>}{5}}\]
%\modalf{\invns^*,\invfc}c)}{>}{5}}
This formula holds for nodes selected by the XPath expression. A correspondence between the main XPath axes over unranked trees and modal formulas over binary trees is given in Figure~\ref{fig:axes-modalities}. In this figure, each logical formula holds for nodes selected by the corresponding XPath axis from a context $\gamma$.
\begin{figure}
  \begin{center}
$\begin{array}{r|l}
  \text{Path} & \text{Logical formula}\\
  \hline
 \gamma/\text{self::*}  &  \gamma  \\
 \gamma/\text{child::*} &  \modalf{\invns^*, \invfc}\gamma  \\
 \gamma/\text{parent::*} &  \modalf{\fc}{\modalf{\ns^*}{\gamma}} \\
 \gamma/\text{descendant::*} &  \modalf{(\invns\mid\invfc)^*, \invfc}\gamma\\
 \gamma/\text{ancestor::*}  &  \modalf{\fc}{\modalf{(\fc\mid\ns)^*}{\gamma}}\\
 \gamma/\text{following-sibling::*} &  \modalf{\invns}{\modalf{\invns^*}{\gamma}}\\
 \gamma/\text{preceding-sibling::*} & \modalf{\ns}{\modalf{\ns^*}{\gamma}} \\
% \psi/\text{following::*} & \\
% \psi/\text{preceding::*} & 
\end{array}$
\end{center}
\caption{XPath axes as modalities over binary trees.}\label{fig:axes-modalities}
\end{figure}

Let consider another example (XPath expression $e_1$):
\begin{verbatim}
child::a/child::b[count(child::e/descendant::h)>3]
\end{verbatim}
Starting from a given context in a tree, this XPath expression navigates to children nodes named ``a'' and selects their children named ``b''. Finally, it retains only those ``b'' nodes for which the qualifier between brackets holds.
The first path can be translated in the logic as follows:
%$$ \vartheta = b \wedge \modalf{\invns^*, \invfc} (a \wedge \modalf{\invns^*, \invfc}\top) $$
$$ \vartheta = b \wedge \ufixpf{\var}{\modalf{\invfc}{(a \wedge \ufixpf{\var^\prime}{\modalf{\invfc}{\top} \vee \modalf{\invns}{\var^\prime}}}) \vee \modalf{\invns}{\var}}$$

The counting part requires a more sophisticated translation in the logic. This is because it makes implicit that ``e'' nodes (whose existence is simply tested for counting purposes) must be children of selected ``b'' nodes. The translation of the full aforementioned XPath expression is as follows:
$$\vartheta \wedge \nominal \wedge \countf{}{(\invfc\mid\invns)^*, (\fc\mid\ns)^*}{\eta}{>}{3}$$  
where $\nominal$ is a new fresh nominal used to mark a ``b'' node which is filtered by the qualifier and the formula $\eta$ describes the counted ``h'' nodes:
%$$\eta = h \wedge \modalf{(\invns\mid\invfc)^*, \invfc}(e \wedge \modalf{\invns^*, \invfc}\nominal)$$
$$\eta = h \wedge \ufixpf{\var}{\modalf{\invfc}{(e \wedge \ufixpf{\var^\prime}{\modalf{\invfc}{\nominal} \vee \modalf{\invns}{\var^\prime}}}) \vee \modalf{\invns}{\var} \vee \modalf{\invfc}{\var}}$$
Intuitively, the general idea behind the translation is to first translate the leading path, use a fresh nominal for marking a node which is filtered, then find at least ``3'' instances of ``h'' nodes from which we can reach back the marked node via the inverse path of the counting formula. 

Since trails make it possible to navigate but not to test properties (like existence of labels), we test for labels in the counted formula $\eta$ and we use a general navigation $(\invfc\mid\invns)^*, (\fc\mid\ns)^*$ to look for counted nodes everywhere in the tree. Introducing the nominal is necessary to bind the context properly (without loss of information). Indeed, the XPath expression $e_1$ makes implicit that a ``e'' node must be a child of a ``b'' node selected by the outer path. Using a nominal, we restore this property by connecting the counted nodes to the initial single context node.

\begin{lemma}
The translation of Core XPath expressions with counting constraints into the logic is linear.
\label{xpathtranslationlemma}
\end{lemma}
It is proven by structural induction in a similar manner to \cite{geneves-pldi07} (in which the translation is proven for expressions without counting constraints). For counting formulas, the use of nominals and the general (constant-size) counting trail make it possible to avoid duplication of trails so that the translation remains linear.

We can now address several decision problems such as equivalence, containment, and emptiness
of XPath expressions. These decision problems are reduced to
test satisfiability for the logic (in the manner of \cite{geneves-pldi07}). We present in Section~\ref{sec:algo} a satisfiability testing algorithm with a single
exponential complexity with respect to the formula size.

In \cite{geneves-pldi07}, it was show the logic is also able to capture XML schema languages.
This allows to test the XPath decision problems in the presence of XML types.
We now show our logic can also succinctly express cardinality constraints on XML types.

\subsection{Regular Tree Languages with Cardinality Constraints}

Regular tree grammars capture most of the schemas in use today \cite{DBLP:journals/toit/MurataLMK05}. The logic can express all regular tree languages (it is easy to prove that regular expression types in the manner of e.g., \cite{DBLP:journals/toplas/HosoyaVP05}  can be linearly translated into the logic: see \cite{geneves-pldi07}). 

In practice, schema languages often provide shorthands for expressing cardinality constraints on node occurrences. XML Schema notably offers two attributes {\em minOccurs} and {\em maxOccurs} for this purpose. For instance, the following XML schema definition:
\begin{Verbatim}[fontsize=\small]
<xsd:element name="a">
  <xsd:complexType>
    <xsd:sequence>
     <xsd:element name="b" minOccurs="4" maxOccurs="9"/>
    </xsd:sequence>
  </xsd:complexType>
</xsd:element>
\end{Verbatim} 
\noindent
is a notation that restricts the number of occurrences of ``\texttt{b}'' nodes to be at least 4 and at most 9, as children of ``\texttt{a}'' nodes. The goal here is to have a succinct notation for expressing regular languages which could otherwise be exponentially large if written with usual regular expression operators. The above regular requirement can be translated as the formula:
\[\phi \wedge \modalf{\fc}{(\countf{}{\ns^\star}{b}{>}{3}} \wedge \countf{}{\ns^\star}{b}{\leq}{9})\]
where $\phi$ corresponds to the regular tree type $a[b^*]$ as follows:
\[\begin{array}{ll}
\phi =&(a \wedge (\neg\modalf{\fc}{\top} \vee \modalf{\fc}{\psi})) \wedge \neg\modalf{\ns}{\top}  \vspace{0.1cm}\\
\psi =&\mu x. \left(b \wedge \neg\modalf{\fc}{\top} \wedge \neg\modalf{\ns}{\top}\right)  \vee \left(b \wedge \neg\modalf{\fc}{\top} \wedge \modalf{\ns}{x}\right)
\end{array}\]

%The logic also supports occurrence constraints on regular expressions (permitted by XML Schema \cite{schema-part1} too).
% For example, consider the regular expression $e=a^{*}ba^{*}$. To require that the regular language denoted by $e$ additionally contains at least two elements $a$, one must express it 
%as $e'=aaa^{*}ba^{*} | aa^{*}baa^{*} | a^{*}baaa^{*}$.  Alternatively, one may write 
% \[\phi_e\wedge\countf{}{\ns^\star}{a}{>}{1} \] where $\phi_e$ is the logical translation of $e$.

This example only involves counting over children nodes. The logic allows counting through more general trails, and in particular arbitrarily deep trails. 
Trails corresponding to the XPath axes ``preceding, ancestor, following'' can be used to constrain the context of a schema. % (see~Figure~\ref{fig:xpath-axes}).
The ``descendant'' trail can be used to specify additional constraints over the subtree defined by a given schema. For instance, suppose we want to forbid webpages containing nested anchors ``$a$'' (whose interpretation makes no sense for web browsers). We can build the logical formula $f$ which is the conjunction of a considered schema for webpages (e.g. XHTML) with the formula $a/\text{descendant::}a$ in XPath notation. Nested anchors are forbidden by the considered schema iff $f$ is unsatisfiable.

As another example, suppose we want paragraph nodes (``$p$'' nodes) not to be nested inside more than 3 unordered lists (``$ul$'' nodes), regardless of the schema defining the context. One may check for the unsatisfiability of the following formula:
\[ p\wedge \countf{}{(\invfc|\invns)^\star,\invfc}{ul}{>}{3} \]

%%%%%%%%%%%%%%%%%%%%%%%%%%%%%%%%%%%%%%%%%%%%%%%%%%%%%%%%%%%%

\section{Satisfiability Algorithm} \label{sec:algo}
We present a tableau-based algorithm for checking satisfiability of formulas. Given a formula, the algorithm seeks to build a tree containing a node selected by the formula. We show that our algorithm is correct and complete: a satisfying tree is found if and only if the formula is satisfiable. We also show that the time complexity of our algorithm is exponential in the size of the formula.

\subsection{Overview}
The algorithm operates in two stages. 

First, a formula $\form$ is decomposed into a set of subformulas, called the \emph{lean}. The lean gathers all subformulas that are useful for determining the truth status of the initial formula, while eliminating redundancies. For instance, conjunctions and disjunctions are eliminated at this stage. More precisely, the lean (defined in \ref{sec:lean}) mainly gathers atomic propositions and modal subformulas. From the lean, one may gather a finite number of formulas, called a $\form\!-\!$node, which may be satisfied at a given node of a tree. Trees of $\form\!-\!$nodes represent the exhaustive search universe in which the algorithm is looking for a satisfying tree.

The second stage of the algorithm consists in the building of sets of such trees in a bottom-up manner, ensuring consistency at each step. Initially, all possible leaves (i.e., $\form\!-\!$node that do not require children nodes) are considered. During further steps, the algorithm considers every possible $\form\!-\!$node that can be connected with a tree of the previous steps, checking for consistency. For instance, if a formula at a $\form\!-\!$node $n$ involve a forward modality  $\modalf{\fc}{\form'}$, then $\form'$ must be verified at the first child of $n$. Reciprocally, due to converse modalities, a $\form\!-\!$node may impose restrictions on its possible parent nodes. The new trees that are built may involve converse modalities, which will be satisfied during further steps of the algorithm. To ensure the algorithm terminates, a bound on the number of times each $\form\!-\!$node may occur in the tree is given.

Finally, the algorithm terminates whenever:
\begin{itemize}
  \item either a tree that satisfies the initial formula has been found, and its root does not contain any pending (unproven) backward modality; or
  \item every tree has been considered (the exploration of the whole search universe is complete): the formula is unsatisfiable.
\end{itemize}

%This approach is very similar to our previous work \cite{geneves-pldi07}, the main challenge is to correctly deal with counting formulas.

\subsection{Preliminaries}
To track where counting formulas are satisfied, we annotate each one with a fresh \emph{counting proposition} $c$, yielding formulas of the form $\countf{c}{\trail}{\form}{\#}{\natn}$.
To define the notions of lean and $\form\!-\!$nodes, we need to extract navigating formulas from counting formulas (Figure~\ref{fig-nav}).
\begin{figure}
\begin{align*}
 nav(\var) &= \var &
 nav(\prop) &= \prop \\
 nav(\top) &= \top &
 nav(c) &= c\\
 nav(\neg \prop) &= \neg \prop &
 nav(\neg \modalf{\modals}{\top})&= \neg \modalf{\modals}{\top}
\end{align*}
\begin{align*}
 nav(\form_1\wedge \form_2)&= nav(\form_1) \wedge nav(\form_2) \\
 nav(\form_1\vee \form_2)&= nav(\form_1) \vee nav(\form_2) \\
 nav(\modalf{\modals}{\form}) &= \modalf{\modals}{nav(\form)} \\
 nav(\ufixpf{\var}{\restrictedform}) &= \ufixpf{\var}{nav(\restrictedform)}\\
 nav(\countf{c}{\trail}{\restrictedform}{>}{\natn})&= \nav{\trail}{\restrictedform\wedge c}\\
 nav(\countf{c}{\trail}{\restrictedform}{\leq}{\natn})&= \nav{\trail}{(\restrictedform \land c) \lor (\neg\restrictedform \land \neg c)}\\
%nav(\neg\countf{}{\trail}{\form}{\leq}{\natn}) &=\nav{\trail}{\form}\\
\nav{\epsilon}{\restrictedform} &= \restrictedform \\
 \nav{\modals}{\restrictedform} &= \modalf{\modals}{\restrictedform}\\
 \nav{\trail_1,\trail_2}{\restrictedform}&= \nav{\trail_1}{\nav{\trail_2}{\restrictedform}} \\
 \nav{\trail_1\mid\trail_2}{\restrictedform} &= \nav{\trail_1}{\restrictedform} \vee \nav{\trail_2}{\restrictedform}\\
 \nav{\trail^\star}{\restrictedform} &= \ufixpf{\var}{nav(\restrictedform)\vee\nav{\trail}{\var}}
\end{align*}
\caption{Navigation extraction from counting formulas}
\label{fig-nav}
\end{figure}

We now define the {\em Fisher-Ladner} relation to extract subformulas. In the following, $i$ ranges over $1$ and $2$.
\begin{align*}
& \fishrel{\form_1\wedge\form_2}{\form_i}, && \fishrel{\form_1\vee\form_2}{\form_i}, \\
& \fishrel{\ufixpf{\var}{\form}}{\form[\subst{\ufixpf{\var}{\form}}{\var}]}, && \fishrel{\countf{c}{\trail}{\restrictedform}{\#}{\natn}}{nav(\countf{c}{\trail}{\restrictedform}{\#}{\natn})},\\
&  \fishrel{\modalf{\modals}\phi}{\phi}.
\end{align*}

The {\em Fisher-Ladner} closure of a formula $\form$, written $\flcl{\form}$, is the set defined as follow.
\begin{align*}
& \flcl{\form}_0 &&=&& \{\form\}, \\
& \flcl{\form}_{i+1} &&=&& \flcl{\form}_i \cup \{\form^\prime \mid \fishrel{\form^{\prime\prime}}{\form^\prime}, \form^{\prime\prime}\in \flcl{\form}_i\},\\
& \flcl{\form}&&=&&\flcl{\form}_k,
\end{align*}
\text{where $k$ is the smallest integer s.t. $\flcl{\form}_k=\flcl{\form}_{k+1}$}.
\noindent
Note that this set is finite since only one expansion of a fixpoint formula is required in order to produce all its subformulas in the closure.

\label{sec:lean}
The {\em lean} of a formula $\form$ is a set of formulas containing navigating formulas of the form $\modalf{\modals}{\top}$, every navigating formulas of the form $\modalf{\modals}{\restrictedform}$ (i.e., that do not contain counting formulas) from $\flcl{\form}$, every proposition occurring in $\form$, written $\Props_{\form}$, every counting proposition, written $C$, and an extra proposition that does not occur in $\form$ used to represent other names, written $p_{\overline\form}$.

\begin{equation*}
\lean{\form}=\{\modalf{\modals}{\top}\}\cup \{\modalf{\modals}{\restrictedform}\in \flcl{\form}\} \cup \Props_{\form} \cup C \cup \{p_{\overline\form}\}
\end{equation*}

A {\em $\form\!-\!$node }, written $\fnode{\form}$, is a subset of $\lean{\form}$, such that:
\begin{itemize}
\item exactly one proposition from $\Props_{\form} \cup \{p_{\overline\form}\}$ is present;
\item when $\modalf{\modals}{\restrictedform}$ is present, then $\modalf{\modals}{\true}$ is present; and
\item both $\modalf{\invfc}{\true}$ and $\modalf{\invns}{\true}$ cannot be present at the same time.
\end{itemize}
The set of $\form\!-\!$nodes is defined as $\fNodes{\form}$.

Intuitively, a node $\fnode{\form}$ corresponds to a formula.
\[
\fnode{\form} = \bigwedge_{\psi \in \fnode{\form}} \psi \wedge \bigwedge_{\psi \in \lean{\form} \setminus \fnode{\form}} \neg \psi
\]

When the formula $\form$ under consideration is fixed, we often omit the superscript.

A \emph{\psitree{\form}} is either the empty tree $\emptyset$, or a triple $(\fnode{\form}, \stree_1,\stree_2)$ where $\stree_1$ and $\stree_2$ are \psitree{\form}s. When clear from the context, we usually refer to \psitree{\form}s simply as trees.

\begin{figure}
  \begin{mathpar}
    \inferrule*{ }{\node \vdash^{\form} \top} \and
    \inferrule*{\restrictedform\in\node}{\node\vdash^{\form} \restrictedform} \and
    \inferrule*{\restrictedform \not\in \node}{\node\vdash^{\form} \neg \restrictedform} \and
    \inferrule*{\node\vdash^{\form} \restrictedform_1 \\ \node\vdash^{\form} \restrictedform_2}{\node\vdash^{\form} \restrictedform_1\wedge\restrictedform_2} \and
    \inferrule*{\node\vdash^{\form} \restrictedform_1}{\node\vdash^{\form} \restrictedform_1\vee\restrictedform_2} \and
    \inferrule*{\node\vdash^{\form} \restrictedform_2}{\node\vdash^{\form} \restrictedform_1\vee\restrictedform_2} \and
    \inferrule*{\node\vdash^{\form} \restrictedform\{\subst{\ufixpf{\var}{\restrictedform}}{\var}\}}{\node\vdash^{\form} \ufixpf{\var}{\restrictedform}}
   % \inferrule*{\node\nvdash^{\form}_{\stree,\pathn} \form^\prime}{\node\vdash^{\form}_{\stree,\pathn} \neg\form^\prime} \and
   % \inferrule*{\form^\prime \in \lean{\form} \\ \form^\prime\notin\node}{\node\nvdash^{\form}_{\stree,\pathn} \form^\prime} \and
%    \inferrule*{
    %\node\vdash^\form\nav{\trail}{\form}
%  }{\node\vdash^{\form}\countf{}{\trail}{\form}{\#}{\natn}} 
%    \inferrule*{\node\nvdash^{\form}_{\stree,\pathn} \form_1 \\ \node\nvdash^{\form}_{\stree,\pathn} \form_2}{\node\nvdash^{\form}_{\stree,\pathn} \form_1\vee\form_2} \and
%    \inferrule*{\node\nvdash^{\form}_{\stree,\pathn} \form_1}{\node\nvdash^{\form}_{\stree,\pathn} \form_1\wedge\form_2} \and
%    \inferrule*{\node\nvdash^{\form}_{\stree,\pathn} \form_2}{\node\nvdash^{\form}_{\stree,\pathn} \form_1\wedge\form_2} \and
%    \inferrule*{\node\nvdash^{\form}_{\stree,\pathn} \form^\prime\{\subst{\ufixpf{\var}{\form^\prime}}{\var}\}}{\node\nvdash^{\form}_{\stree,\pathn} \ufixpf{\var}{\form^\prime}}\and
%    \inferrule*{\node\vdash^{\form}_{\stree,\pathn} \form^\prime}{\node\nvdash^{\form}_{\stree,\pathn} \neg\form^\prime}
  \end{mathpar}
  \caption{Local entailment relation: between nodes and formulas}
  \label{fig:entailmentnode}
\end{figure}

We now turn to the definition of consistency of a \psitree{\form}. To this end, we define an entailment relation between a node and a formula in Figure~\ref{fig:entailmentnode}.
% The tree $\stree$ and path $\pathn$ annotations are used only when dealing with counting formulas, which are checked only when finished candidate \psitree{\form}s are generated (see below). We omit these annotations when they are not used.

  Two nodes $\node_1$ and $\node_2$ are consistent under modality $\modals \in \{\fc,\ns\}$, written $\fbrel{\form}{\node_1}{\modals}{\node_2}$, iff
  \begin{align*}
    \forall \modalf{\modals}\psi \in \lean{\form}&,\modalf{\modals}\psi \in \node_1 \iff \node_2 \vdash^{\form} \psi\\
    \forall \modalf{\dual\modals}\psi \in \lean{\form}&, \modalf{\dual\modals}\psi \in \node_2 \iff \node_1 \vdash^{\form} \psi
  \end{align*}

Consistency is checked each time a node is added to the tree, ensuring that forward modalities of the node are indeed satisfied by the nodes below, and that pending backward modalities of the node below are consistent with the added node. Note that counting formulas are not considered at this point, as they are globally verified in the next step.

%In order to define a  consistent entailment relation for counting formulas in \psitree{\form}s, we need to consider several occurrences of the same nodes. 
%Hence, given a formula $\form$, we define the following multisets.
%\begin{align*}
%& \fNodes{\form}_0 &&=&& \{\node &&\mid&& \text{ $\node$ is a $\form\!-\!$node}\} \\
%& \fNodes{\form}_{i+1} &&=&& \fNodes{\form}_i \uplus \biguplus^{k}_{j=1}\{\node &&\mid&& \fishrel{\form^\prime}{\countf{}{\form_1}{\trail}{\form_2}{\#}{k}},\form^\prime\in\flcl{\form}_i, \\
%        &&&&& &&&&  \node \vdash^\form \nav{\trail}{\form_2}, \node\in \fNodes{\form}_i \}\\
%& \fNodes{\form} &&=&& \fNodes{\form}_k, &&&&\text{ where $k$ is the smallest integer s.t. }\flcl{\form}_k=\flcl{\form}_{k+1}
%\end{align*}

Upon generation of a finished tree, i.e., a tree with no pending backward modality, one may check whether a node of this tree satisfies $\form$. To this end, we first define forward navigation in a \psitree{\form} $\stree$. Given a path consisting of forward modalities $\pathn$, $\stree(\pathn)$ is the node at that path. It is undefined if there is no such node.
\begin{align*}
  (\node,\stree_1,\stree_2)(\epsilon) &= \node \\
  (\node,\stree_1,\stree_2)(\fc\pathn) &= \stree_1(\pathn) \\
  (\node,\stree_1,\stree_2)(\ns\pathn) &= \stree_2(\pathn)
\end{align*}
We also allow extending the path with backward modalities if they match the last modality of the path.
\begin{align*}
  (\node,\stree_1,\stree_2)(\pathn \fc \invfc) &= (\node,\stree_1,\stree_2)(\pathn) \\
  (\node,\stree_1,\stree_2)(\pathn \ns \invns) &= (\node,\stree_1,\stree_2)(\pathn)
\end{align*}

Now, we are able to define an entailment relation along paths in
\psitree{\form}s in Figure~\ref{fig:countingentailment}. This relation extends
local entailment relation (Figure~\ref{fig:entailmentnode}) with checks for
counting formulas. Note that the case for fixpoints is contained in the case for formulas with no counting subformula. In the ``less than'' case, we need to make sure that every node reachable through the trail is taken into account, either as counted if it satisfies $\restrictedform$, or not counted otherwise  (in this case, $\neg \restrictedform$ denotes the negation normal form). 

\begin{figure}
  \begin{mathpar}
   % \inferrule*{\node^\prime \vdash^{\form}_{\stree,\pathn \modals} \form^\prime \\ \stree(\pathn \modals) = \node^\prime}{\node \vdash^{\form}_{\stree,\pathn} \modalf{\modals}\form^\prime} \and
   \inferrule*{\form^\prime \text{ does not contain counting formulas} \\ \stree(\pathn) \vdash^\form \form^\prime}{\pathn \vdash^{\form}_{\stree} \form^\prime}\and
    \inferrule*{\pathn\vdash^{\form}_{\stree}  \form_1 \\ \pathn\vdash^{\form}_{\stree}  \form_2}{\pathn\vdash^{\form}_{\stree}  \form_1\wedge\form_2} \and
    \inferrule*{\pathn\vdash^{\form}_{\stree}  \form_1}{\pathn\vdash^{\form}_{\stree}  \form_1\vee\form_2} \and
    \inferrule*{\pathn\vdash^{\form}_{\stree}  \form_2}{\pathn\vdash^{\form}_{\stree}  \form_1\vee\form_2} \and
    \inferrule*{\pathn \modals \vdash^{\form}_{\stree}\form^\prime}{\pathn \vdash^{\form}_{\stree} \modalf{\modals}\form^\prime}\and
    \inferrule*{
     | \{ \node^\prime,\; \pathn^\prime \in \trail \wedge \stree(\pathn \pathn^\prime) = \node^\prime \wedge \node^\prime \vdash^{\form} \restrictedform \land c \} | > \natn
}{ \pathn \vdash^{\form}_{\stree} \countf{c}{\trail}{\restrictedform}{>}{\natn}}
\and
    \inferrule*{
     | \{ \node^\prime,\; \pathn^\prime \in \trail \wedge \stree(\pathn \pathn^\prime) = \node^\prime \wedge \node^\prime \vdash^{\form} \restrictedform \land c \} | \leq \natn \\
     \forall \pathn^\prime \in \trail, 
     \stree(\pathn\pathn^\prime) \vdash^\form (\restrictedform \land c) \lor (\neg \restrictedform \land \neg c)
}{ \pathn \vdash^{\form}_{\stree} \countf{c}{\trail}{\restrictedform}{\leq}{\natn}}   \end{mathpar}
  \caption{Global entailment relation (incl. counting formulas)}
  \label{fig:countingentailment}
\end{figure}

We conclude these preliminaries by introducing some final notations.
The {\em root} of a \psitree{\form} is defined as follows. 
\begin{align*}
\stroot{\emptyset} & = \emptyset\\
\stroot{(\node,\stree_1,\stree_2)} & = \node
\end{align*}

A \psitree{\form} $\stree$ {\em satisfies} a formula $\form$, written $\stree\vdash \form$, if neither $\modalf{\invfc}\top$ nor $\modalf{\invns}\top$ occur in $\stroot{\stree}$, and if there is a path $\pathn$ such that $\pathn\vdash^{\form}_{\stree} \form$. A set of trees $ST$ {\em satisfies} a formula $\form$, written $ST\vdash \form$, when there is a tree $\stree\in ST$ such that
$\stree\vdash \form$.

%%%%%%%%%%%%%%%%%%%%%%%%%%%%%%%%%%%%%%%%%%%%%%%%%%%%%%%%%%%%%%%%
\subsection{The Algorithm}
We are now ready to present the algorithm, which is parameterized by $K(\form)$ (defined in Figure~\ref{fig:boundK}), the maximum number of occurrences of a given node in a path from the root of the tree to a leaf. The algorithm builds consistent candidate trees from the bottom up, and checks at each step if one of the built tree satisfies the formula, returning $1$ if it is the case. As the set of nodes from which to build the trees is finite, it eventually stops and returns $0$ if no satisfying tree has been found.

\begin{algorithm}
\caption{Check Satisfiability of $\form$} 
\label{satalgo}
\begin{algorithmic}
    \STATE $ST \leftarrow \emptyset$
    \REPEAT
        \STATE $AUX \leftarrow \{(\node,\stree_1,\stree_2) \mid$ ~~\quad\qquad \textcolor{gray}{\COMMENT{we extend the trees}}\\
        \quad $\nmax(\node,\stree_1,\stree_2) \leq K(\form) + 2$ ~\textcolor{gray}{\COMMENT{with an available node}}\\
        \quad for $i$ in $\fc,\ns$                               ~~~\qquad\qquad\qquad \textcolor{gray}{\COMMENT{and each child is either}}\\
        \quad $\stree_i = \emptyset$ and $\modalf{i}{\top} \notin \node$ \qquad\qquad\textcolor{gray}{\COMMENT{an empty tree}}\\
        \quad or $\stree_i\in ST$ ~~\quad\qquad\qquad\qquad\textcolor{gray}{\COMMENT{or a previously built tree}}\\
        \quad\quad $\modalf{\dual{i}}\top \in \stroot{\stree_i}$ ~~\textcolor{gray}{\COMMENT{with pending backward modalities}}\\
        \quad\quad $\fbrel{\form}{\node}{i}{\stroot{\stree_i}}\}$ \qquad\textcolor{gray}{\COMMENT{checking consistency}}
        \IF{$AUX \subseteq ST$}
      \RETURN{$0$} \qquad\qquad\qquad\qquad\textcolor{gray}{\COMMENT{No new tree was built}}
        \ENDIF
      \STATE $ST \leftarrow ST \cup AUX$
    \UNTIL{$ST \vdash \form$}
    \RETURN{$1$}
\end{algorithmic}
\end{algorithm}

To bound the size of the trees that are built, we restrict the number of identical nodes on a path from the root to any leaf by $K(\form) + 2$, defined in Figure \ref{fig:boundK}, using $\nmax$ defined as follows.
\begin{align*}
  \nmax(\node,\stree_1,\stree_2) &= \imax(\nmax(\node,\stree_1),\nmax(\node,\stree_2))\\
  \nmax(\node,(\node,\stree_1,\stree_2)) & = 1 + \nmax(n,\stree_1,\stree_2)\\
  \nmax(\node,(\node^\prime,\stree_1,\stree_2)) & = \nmax(n,\stree_1,\stree_2) \quad\text{if $\node \neq \node^\prime$}\\
  \nmax(\node, \emptyset) &= 0
\end{align*}

\begin{figure}
\begin{align*}
& K(\prop) = K(\neg\prop) = K(\neg \modalf{\modals}{\true}) = K(\true) = K(\ufixpf{\var}{\restrictedform})=0 \\
& K(\form_1\wedge\form_2) = K(\form_1\vee\form_2)= K(\form_1)+K(\form_2) \\
& K(\modalf{\modals}{\form}) = K(\form) \\
& K(\countf{}{\trail}{\restrictedform}{\#}{\natn}) = \natn+1
\end{align*}
\caption{Occurrences bound}
\label{fig:boundK}
\end{figure}

Consider for instance the formula $\form = \prop_1\wedge \modalf{\fc}{\countf{}{\ns^\star}{\prop_2}{>}{2}}$. The computed lean is as follows, where $\psi = \ufixpf{\var}{(\prop_2 \land c) \vee \modalf{\ns}{\var}} $.
\[\{\prop_1,\prop_2,\prop_3,c,\modalf{\fc}{\top},\modalf{\ns}{\top},\modalf{\invfc}{\top}, \modalf{\invns}{\top}, \modalf{\fc}{\psi}, \modalf{\ns}{\psi} \}\]

Names other than $\prop_1$ and $\prop_2$ are represented by $\prop_3$; $c$ identifies counted nodes.
Computing the bound on nodes, we get $K(\form) = 3$.

After the first step, $ST$ consists of the trees $(\{\prop_i\}, \emptyset, \emptyset)$, $(\{\prop_i,c\}, \emptyset, \emptyset)$, $(\{\prop_i, \modalf{\dual{j}}{\top}\}, \emptyset, \emptyset)$, and $(\{\prop_i, c,  \modalf{\dual{j}}{\top}\}, \emptyset, \emptyset)$ with $i \in \{1,2,3\}$ and $j \in \{\fc,\ns\}$. At this point the three finished trees in $ST$ are tested and found not to satisfy $\form$.

After the second iteration many trees are created, but the one of interest is the following.
\[\tree_0 = (\{\prop_2,c,\modalf{\ns}{\top},\modalf{\invns}{\top}, \modalf{\ns}{\psi}\},\emptyset,(\{\prop_2,c,\modalf{\invns}{\top}\},\emptyset,\emptyset))\]

The third iteration yields the following tree.
\[\tree_1 = (\{\prop_2,c,\modalf{\ns}{\top},\modalf{\invfc}{\top}, \modalf{\ns}{\psi}\},\emptyset,\tree_0)\]

We can conclude by the fourth iteration when we find the tree $(\{\prop_1, \modalf{\fc}{\psi}, \modalf{\fc}{\top}\},\tree_1,\emptyset)$, which is found to satisfy $\form$ at path $\epsilon$. 
As the nodes at every step are different, the limit is not reached. 
Figure \ref{algfig} depicts a graphical representation of the example
where counted nodes (containing $c$) are drawn as thick circles.

\begin{figure}
\centering
\begin{tikzpicture}[scale=0.6]. 
\draw [thin,dotted] (-1,3) -- (9,3);
\draw [thin,dotted] (-1,5) -- (9,5);
\draw [thin,dotted] (-1,7) -- (9,7);
\node (p1) at (0,2) [circle,draw] {$\prop_1$};
\node (p2) at (2,2) [circle,draw] {$\prop_2$};
\node (p3) at (4,2) [circle,draw] {$\prop_3$};
%\node (p4) at (4,2)  {$\ldots$};
\node (p10) at (5,2)  {$\ldots$};
\node (p5) at (7,2) [circle,draw,very thick] {$\prop_2$};
\node (p6) at (6,2)  {$\ldots$};
\node (p7) at (6,4) [circle,draw,very thick] {$\prop_2$};
\node (p8) at (5,6) [circle,draw,very thick] {$\prop_2$};
\node (p9) at (6,8) [circle,draw] {$\prop_1$};
\draw [->] (p9) to node[auto] {$\fc$} (p8);
\draw [->] (p8) to node[auto] {$\ns$} (p7);
\draw [->] (p7) to node[auto] {$\ns$} (p5);
\end{tikzpicture}
\caption{Checking  $\form = \prop_1\wedge \modalf{\fc}{\countf{}{\ns^\star}{\prop_2}{>}{2}}$} \label{algfig}
\end{figure}

\subsection{Termination}
Proving termination of the algorithm is straightforward, as only a finite number of trees may be built and the algorithm stops as soon as it cannot build a new tree.

\subsection{Soundness}

If the algorithm terminates with a candidate, we show that the initial formula is satisfiable. Let $\stree$ and $\pathn$ be the \psitree{\form} and path such that $\pathn \vdash^\form_\stree \form$. We build a tree from $\stree$ and show that the interpretation of $\form$ for this tree includes the node at path $\pathn$.

We write $\tree(\stree)$ for the tree $(\Nodes,\brel,\tlabel)$ defined as follows. We first rewrite $\Gamma$ such that each node $\node$ is replaced by the path to reach it (i.e, nodes are identified by their path).
\begin{align*}
  path(\node,\stree_1,\stree_2) & \rightarrow (\epsilon, path(\fc,\stree_1), path(\ns,\stree_2))\\
  path(\pathn, (\node,\stree_1,\stree_2)) & \rightarrow (\pathn, path(\pathn \fc,\stree_1), path(\pathn \ns,\stree_2))\\
  path(\pathn, \emptyset) & \rightarrow \emptyset
\end{align*}

We then define:
\begin{itemize}
\item $\Nodes = \stnodes{path(\stree)}$;
\item for every $(\pathn,\stree_1,\stree_2)$ in $path(\stree)$ and $i=\fc,\ns$,
 if $\stree_i\neq\emptyset$ then $R(\pathn,i)=\pathn i$ and $R(\pathn i, \dual{i}) = \pathn$; and
\item for all $\pathn\in\Nodes$ if $\prop\in\stree(\pathn)$ then $\tlabel(\pathn)=\prop$.
\end{itemize}

\begin{lemma}\label{lem:localsound}
  Let $\restrictedform$ a subformula of $\form$ with no counting formula. If $\stree(\pathn) \vdash^\form \restrictedform$ then we have $\pathn \in \semf{\restrictedform}{\tree(\stree)}{\emptyset}$.
\end{lemma}

\begin{proof}
  We proceed by induction on the lexical ordering of the number of unfolding of $\restrictedform$ that are required for $\tree(\stree)$ as defined by Lemma~\ref{lem:finite_unfolding}, and of the size of the formula.
  
  The base cases are $\true$, atomic or counting propositions, and negated forms. These are immediate by definition of $\semf{\restrictedform}{\tree(\stree)}{\emptyset}$. The cases for disjunction and conjunction are immediate by induction (the formula is smaller). The case for fixpoints is also immediate by induction, as the number of unfoldings required decreases, and as $\semf{\ufixpf{\var}{\restrictedform}}{\tree(\stree)}{\emptyset} = \semf{\restrictedform\{\subst{\ufixpf{\var}{\restrictedform}}{\var}\}}{\tree(\stree)}{\emptyset}$.
  
  The last case is the presence of a modality $\modalf{\modals}\restrictedform$ from the $\form$node $\stree(\pathn)$. In this case we rely on the fact that the nodes $\stree(\pathn \modals)$ and $\stree(\pathn)$ are consistent to derive $\stree(\pathn \modals) \vdash^\form \restrictedform$. We then conclude by induction as the formula is smaller.
\end{proof}

\begin{theorem}[Soundness]\label{thm:soundness} If
$\pathn \vdash^\form_\stree \form$ then $\pathn \in \semf{\form}{\tree(\stree)}{\emptyset}$
\end{theorem}

\begin{proof}
  The proof proceeds by induction on the derivation of $\pathn \vdash^\form_\stree \form$. Most cases are immediate (or rely on Lemma \ref{lem:localsound}). For the ``greater than'' counting case, we rely on the $k+1$ selected nodes that have to satisfy $\restrictedform \land c$ thus $\restrictedform$. In addition, in the ``less than'' case, every node that is not counted has to satisfy $\neg \restrictedform \land \neg c$, so in particular $\neg \restrictedform$. In both cases we conclude by induction.
\end{proof}

\subsection{Completeness}

Our proof proceeds in two step. We build a \psitree{\form} that satisfies the
formula, then we show it is actually built by the algorithm. As the proof is quite complex, we devote some space to detail it.

Assume that formula $\form$ is satisfiable by a tree
$\tree$. We consider the smallest such tree (i.e., the tree with the fewest
number of nodes) and fix $\node^\star$, a node witnessing satisfiability.

We now build a \psitree{\form} homomorphic to $\tree$, called the lean labeled
version of $\form$, written $\stree(\tree,\form)$. To this end, we start by annotating counted nodes along with their corresponding counting proposition, yielding a new tree $\tree_c$. Starting from $\node^\star$ and by induction on $\form$, we proceed as follows. For formulas with no counting subformula, including recursion, we stop. For conjunction and disjunction of formulas, we recursively annotate according to both subformulas. For modalities, we recursively annotate from the node under the modality. For $\countf{c}{\trail}{\restrictedform}{\leq}{\natn}$, we annotate every selected node with the counting proposition corresponding to the formula. For $\countf{c}{\trail}{\restrictedform}{>}{\natn}$, we annotate exactly $\natn+1$ selected nodes.

We now extend the semantics of formulas to take into account counting propositions and annotated nodes, written $\semfc{\cdot}{\tree}{\valuation}$. The definition is identical to Figure \ref{formsem}, with one addition and two changes. The addition is for counting propositions, which we define as $\node \in \semfc{c}{\tree}{\valuation}$ iff $\node$ is annotated by $c$. The two changes are for counting propositions, which we define as follows, where we select only nodes that are annotated.

\begin{align*}
\semfc{\countf{}{\alpha}{\form^\prime}{\leq}{\natn}}{\tree}{\valuation} &=
    \{\node , |\{\node^\prime\in \semfc{\form^\prime}{\tree}{\valuation} \cap \semfc{c}{\tree}{\valuation} , \path{\node}{\alpha}{\node^\prime}\}|\leq  \natn \}\\
\semfc{\countf{}{\alpha}{\form^\prime}{>}{\natn}}{\tree}{\valuation} &=
    \{\node ,  |\{\node^\prime\in \semfc{\form^\prime}{\tree}{\valuation} \cap \semfc{c}{\tree}{\valuation} , \path{\node}{\alpha}{\node^\prime}\} |>  \natn \}
\end{align*}

We show that this modification of the semantics does no change the satisfiability of the formula.

\begin{lemma}\label{lem:annotated_satisfiability}
  We have $\node^\star \in \semfc{\form}{\tree}{\emptyset}$.
\end{lemma}

\begin{proof}
  We proceed by recursion on the derivation $\node^\star \in \semf{\form}{\tree}{\emptyset}$. The cases where no counting formula is involved, thus including fixpoints, are immediate, as the selected nodes are identical. The disjunction, conjunction, and modality cases are also immediate by induction. The interesting cases are the counting formulas.
  
  For $\countf{c}{\trail}{\restrictedform}{>}{\natn}$, as there are exactly $\natn+1$ nodes annotated, the property is true by induction. For $\countf{c}{\trail}{\restrictedform}{\leq}{\natn}$, we rely on the fact that every counted node is annotated. We conclude by remarking that $\restrictedform$ does not contain a counting formula, thus we have $\semfc{\restrictedform}{\tree}{\valuation} = \semf{\restrictedform}{\tree}{\valuation}$ and $\semfc{\neg \restrictedform}{\tree}{\valuation} = \semf{\neg \restrictedform}{\tree}{\valuation}$.
\end{proof}

To every node $\node$, we associate $\fnode{\form}$, the largest subset of formulas of the lean selecting the node.
\begin{equation*}
\fnode{\form} = \{\form_0\mid \node \in \semf{\form_0}{\tree}{\emptyset} ,\form_0\in\lean{\form}\}
\end{equation*}

This is a $\form$-node as it contains one and exactly one proposition, and if it includes a modal formula $\modalf{\modals}{\psi}$, then  it also includes $\modalf{\modals}{\top}$. The tree $\stree(\tree,\form)$ is then built homomorphically to $\tree$.

In the remainder of this section, we write $\stree$ for $\stree(\tree,\form)$.
We now check that $\stree$ is consistent, starting with local consistency.

\begin{figure}
  \begin{mathpar}
    \inferrule*{\psi \in \lean{\form}}{\psi \induced \lean{\form}}\and
    \inferrule*{\psi_1 \induced \lean{\form} \\ \psi_2 \induced \lean{\form}}{\psi_1 \land \psi_2 \induced \lean{\form}} \and
    \inferrule*{\psi_1 \induced \lean{\form} \\ \psi_2 \induced \lean{\form}}{\psi_1 \lor \psi_2 \induced \lean{\form}} \and
    \inferrule*{ }{\true \induced \lean{\form}} \and
    \inferrule*{\psi \in (\Props_{\form} \cup \modalf{\modals}\true \cup C) }{\neg\psi \induced \lean{\form}}
  \end{mathpar}
  \caption{Formula induced by a lean}
  \label{fig:induced_by_lean}
\end{figure}

In the following, we say a formula $\psi$ is induced by the lean of $\form$,
written $\psi \induced \lean{\form}$, if it consists of the boolean combination of subformulas from the lean as defined in
Figure~\ref{fig:induced_by_lean}.

\begin{lemma}\label{lem:formulas_induced}
  Let $\modalf{\modals}\psi$ be a formula in $\lean{\form}$, and let $\psi^\prime$ be $\psi$ after unfolding its fixpoint formulas not under modalities. We have $\psi^\prime \induced \lean{\form}$.
\end{lemma}

\begin{proof}
  By definition of the lean and of the $\induced$ relation.
\end{proof}

\begin{lemma}\label{lem:complete_entailment}
  Let $\psi$ be a formula induced by $\lean{\form}$. We have $\node \in \semfc{\psi}{\tree}{\emptyset}$ if and only if $\fnode{\form} \vdash^{\form} \psi$.
\end{lemma}

\begin{proof}
   We proceed by induction on $\psi$. The base cases (the formula is in the $\form$-node or is a negation of a lean formula not in the $\form$-node) hold by definition of $\fnode{\form}$. The inductive cases are straightforward as these formulas only contain fixpoints under modalities.
\end{proof}

\begin{lemma}\label{lem:complete_consistent}
  Let $\node_1$ and $\node_2$ such that $\brel(\node_1,\modals) = \node_2$ with $\modals \in \{\fc,\ns\}$. We have $\fbrel{\form}{\fnode{\form}_1}{\modals}{\fnode{\form}_2}$.
\end{lemma}

\begin{proof}
  Let $\modalf{\modals}\psi$ be a formula in $\lean{\form}$. We show that $\modalf{\modals}\psi \in \fnode{\form}_1 \iff \fnode{\form}_2 \vdash^\form \psi$. We have $\modalf{\modals}\psi \in \fnode{\form}_1$ if and only if 
  $\node_1 \in \semfc{\modalf{\modals}\psi}{\tree}{\emptyset}$ by definition of $\fnode{\form}_1$, which in turn holds if and only if $\node_2 = \brel(\node_1,\modals) \in \semfc{\psi}{\tree}{\emptyset}$. 
We now consider $\psi^\prime$ which is $\psi$ after unfolding its fixpoint formulas not under modalities. We have $\semfc{\psi^\prime}{\tree}{\emptyset} = \semfc{\psi}{\tree}{\emptyset}$ and we conclude by Lemmas \ref{lem:formulas_induced} and \ref{lem:complete_entailment}.
\end{proof}

We now turn to global consistency, taking counting formulas into account.

\begin{lemma}\label{comple1}
  Let $\form_s$ be a subformula of $\form$, and $\pathn$ be a path from the root in $\tree$ such that $\tree(\pathn) \in \semfc{\form_s}{\tree}{\emptyset}$. We then have $\pathn \vdash^{\form}_{\stree} \form_s$.
\end{lemma}

\begin{proof}
  We proceed by induction on $\form_s$.
  
  If $\form_s$ does not contain any counting formula, we consider $\form_s^\prime$ which is $\form_s$ after unfolding its fixpoint formulas not under modalities. We have $\semfc{\form_s^\prime}{\tree}{\emptyset} = \semfc{\form_s}{\tree}{\emptyset}$ and $\form_s^\prime \induced \lean{\form}$. We conclude by Lemma \ref{lem:complete_entailment}.
  
  For most inductive cases, the proof is immediate by induction, as  the formula size decreases.
  
  For $\countf{c}{\trail}{\restrictedform}{>}{\natn}$, we have by induction for every counted node $\pathn\pathn^\prime \vdash^\form_\stree \restrictedform$ and $\pathn\pathn^\prime \vdash^\form_\stree c$. We conclude by the conjunction rule and by the counting rule of Figure \ref{fig:countingentailment}.
    
  For $\countf{c}{\trail}{\restrictedform}{\leq}{\natn}$, we proceed as above for the counted nodes. For the nodes that are not counted, we have $\stree(\pathn\pathn^\prime) \vdash^\form \neg\restrictedform$ by Lemma \ref{lem:complete_entailment} (since $\neg\restrictedform \induced \lean{\form}$). We conclude by remarking that the node is not annotated by $c$, hence $\stree(\pathn\pathn^\prime) \vdash^\form \neg c$.
\end{proof}

We next show that the \psitree{\form} $\stree$ is actually built by the algorithm. The proof follows closely the one from \cite{geneves-pldi07}, with a crucial exception: we need to make sure there are enough instances of each formula. Indeed, in \cite{geneves-pldi07}, the algorithm uses a $\form$type (a subset of $\lean{\form}$) at most once on each branch from the root to a leaf of the built tree. This yields a simple condition to stop the algorithm and conclude the formula is unsatisfiable. However, in the presence of counting formulas, a given $\form$type may occur more than once on a branch. To maintain the termination of the algorithm, we bound the number of identical $\form$type that may be needed by $K(\form)$ as defined in Figure \ref{fig:boundK}. We thus need to check that this bound is sufficient to build a tree for any satisfiable formula.

We recall that $\form$ is a satisfiable formula and $\tree$ is a smallest tree such that $\form$ is satisfied, and $\node^\star$ is a witness of satisfiability.

We proceed in two steps: first we show that counted nodes (with counted propositions) imply a bound on the number of identical $\form$types on a branch for a smallest tree. Second, we show that this minimal marking is bound by $K(\form)$.

In the following, we call counted nodes and node $\node^\star$ \emph{annotations}. We define the \emph{projection} of an annotation on a path. Let $\pathn$ be a path from the root of the tree to a leaf. An annotation projects on $\pathn$ at $\pathn_1$ if $\pathn = \pathn_1 \pathn_2$, the annotation is at $\pathn_1 \pathn_m$, and $\pathn_2$ shares no prefix with $\pathn_m$.

\begin{lemma}\label{lem:annotation_present}
  Let $\stree^\prime$ be the annotated tree, $\pathn$ a path from the root of the tree to a leaf, $\node_1$ and $\node_2$ two distinct nodes of $\pathn$ such that $\fnode{\form}_1 = \fnode{\form}_2$. Then either annotations projects both on $\pathn$ at $\node_1$ and $\node_2$, or an annotation projects strictly between $\node_1$ and $\node_2$.
\end{lemma}

\begin{proof}
  We proceed by contradiction: we assume there is no annotation that projects between $\node_1$ and $\node_2$ and at most one of them has an annotation that projects on it. Without loss of generality, we assume that $\node_2$ is below $\node_1$ in the tree.
  
  Assume neither $\node_1$ nor $\node_2$ is annotated (through projection). We consider the tree $\stree_s$ where $\node_2$ is ``grafted'' upon $\node_1$. Formally, let $\pathn_1$ be the path to $\node_1$ and $\pathn_1\pathn_2$ the path to $\node_2$. We remove every node whose path is of the form $\pathn_1\pathn_3$ where $\pathn_2$ is not a prefix of $\pathn_3$, and we also remove node $\node_2$. The mapping $\brel^\prime$ from nodes and modalities to nodes is the same as before for the node that are kept except for $\node_1$, where $\brel^\prime(\node_1,\fc) = \brel(n_2,\fc)$ and $\brel^\prime(\node_1,\ns) = \brel(\node_2,\ns)$. For every path $\pathn$ of $\stree$, let $\pathn_s$ be the potentially shorter path if it exists (i.e., if it was not removed when pruning the tree). More precisely, if $\pathn^\prime = \pathn^\prime_1 \pathn^\prime_3$ where $\pathn^\prime_1$ is a prefix of $\pathn_1$ and the paths are disjoint from there, then $\stree_s(\pathn^\prime) = \stree(\pathn^\prime)$. If $\pathn^\prime = \pathn_1 \pathn_2 \pathn_3$, then $\stree_s(\pathn_1\pathn_3) = \stree(\pathn^\prime)$.
  
  We now show that $\stree_s$ still satisfies $\form$ at $\node^\star$, a contradiction since this tree is strictly smaller than $\stree$.
  
  First, as there was no annotation projected, $\node^\star$ is still part of this tree at a path $\pathn_s$. We show that we have $\pathn_s \vdash^{\form}_{\stree_s} \form$ by induction on the derivation $\pathn \vdash^{\form}_{\stree} \form$. Let $\pathn^\prime \vdash^{\form}_{\stree} \form^\prime$ in the derivation, assuming that $\pathn^\prime_s$ is defined.
  
  The case where $\form^\prime$ does not mention any counting formula is trivial: $\stree(\pathn^\prime) = \stree_s(\pathn^\prime_s)$ thus local entailment is immediate.
  
  Conjunction and disjunction are also immediate by induction.
  
We now turn to the modality case, $\modalf{\modals} \form^\prime$ where $\form^\prime$ contains a counting formula. If $\pathn^\prime$ is neither $\pathn_1$ nor $\pathn_1\pathn_2$, we deduce from the fact that $\pathn^\prime_s$ is defined that $(\pathn^\prime \modals)_s$ is also defined and we conclude by induction. We now assume that $\pathn^\prime$ is either $\pathn_1$ or $\pathn_1\pathn_2$ and find a contradiction. First, remark that $\pathn^\prime \vdash^{\form}_{\stree} \modalf{\modals} \form^\prime$ implies that the navigation generated by $\modalf{\modals} \form^\prime$ is in $\stree(\pathn_1) = \stree(\pathn_1 \pathn_2)$. As each syntactic occurrence of a counting formula mentions a distinct counting proposition $c$, this is possible only if the counting formula is under a fixpoint or under another counting formula, both of which are impossible.
    
We finally turn to the counting case $\countf{c}{\trail}{\restrictedform}{\#}{\natn}$. We say that a path \emph{does not cross over} when this path does not contain $\node_1$ nor $\node_2$. For nodes that are reached using paths that do not cross over, we conclude by induction that they are also counted. We show that the remaining nodes reached through a crossover remain reachable (there cannot be any counted node in the part of the tree that is removed since counted nodes are annotated and there was no annotation in the part removed). Without loss of generality, assume that $\pathn^\prime$ is a prefix of $\pathn_1$ (the counting formula is in the ``top'' part of the tree), and let $\pathn_n$ be the path from the counting formula to the counted node ($\pathn_n$ is an instance of the trail $\trail$). This path is of the shape $\pathn^\prime_1 \pathn_2 \pathn_c$, with $\pathn_1 = \pathn^\prime \pathn^\prime_1$. We now show that the path $\pathn^\prime_1 \pathn_c$ is an instance of $\trail$ if and only if $\pathn_n$ is an instance of the trail, thus the same node is still reached.

Recall that $\trail$ is of the shape $\trail_1, \ldots, \trail_n, \trail_{n+1}$ where $\trail_1$ to $\trail_n$ are of the form $\trail_{r_i}^\star$ and where $\trail_{n+1}$ does not contain a repeated trail. We say that a prefix $\pathn_p$ of a path $\pathn$ \emph{stops at $i$} if there is a suffix $\pathn_s$ such that $\pathn_p\pathn_s$ is still a prefix of $\pathn$, $\pathn_p\pathn_s \in \trail_1, \ldots, \trail_i$, and there is no shorter suffix $\pathn^\prime_s$ and $j$ such that $\pathn_p\pathn^\prime_s \in \trail_1, \ldots, \trail_j$. (Intuitively, $\trail_i$ is the trail being used when matching the end of $\pathn_p$.) If there are several satisfying indices $i$, we consider the smallest.

We first show that a counting proposition is necessarily mentioned in a formula of $\fnode{\form}_2$, by contradiction. Assume no counting proposition is mentioned, yet the counting crossed-over. This can only occur for a ``less than'' counting formula that reaches $n_2$ which is not counted (because the formula was false), and if there is no path whose $\pathn_n$ is a strict prefix that is an instance of $\trail$ (otherwise, by definition of the lean and of $nav$ (Figure \ref{fig-nav}), a formula of the form $\nav{\trail^\prime}{(\restrictedform \land c) \lor (\neg\restrictedform \land \neg c)}$ would be true and thus would be present, contradicting the assumption that no counting proposition is mentioned). Since $\fnode{\form}_1 = \fnode{\form}_2$, the same is true for $\fnode{\form}_1$, a direct contradiction to the fact that $n_2$ is also reached by the trail. Thus counting propositions are mentioned in $\fnode{\form}_1$ and $\fnode{\form}_2$.

We next show that there are $i \leq j \leq n$ such that both $\pathn^\prime_1$ stops at $i$ and $\pathn^\prime_1 \pathn_2$ stop at $j$, i.e., neither $i$ nor $j$ may be $n+1$. Recall that $\trail_{n+1}$ does not contain a repeated subtrail. Thus every formula of $\fnode{\form}_2$ mentioning $c$ is of the form $\nav{\trail^\prime}{\restrictedform}$, where $\trail^\prime$ does not contain a repetition. We consider the largest such formula. Since $n_1$ is before $n_2$ in the path from the counting node to the counted node, a similar formula with a larger trail or with a repetition must occur in $\fnode{\form}_1$, contradicting $\fnode{\form}_1 = \fnode{\form}_2$.

Consider next the suffixes $\pathn_s^1$ and $\pathn_s^2$ computed when stating that the paths stop at $i$ and $j$. These suffixes correspond to the path matching the end of $\trail_i$ and $\trail_j$, respectively (before the next iteration or switching to the next subtrail). They have matching formulas in $\fnode{\form}_1$ and $\fnode{\form}_2$. As the formulas are present in both nodes, then the remainder of the paths ($\pathn_2\pathn_c$ and $\pathn_c$) are instances of $(\pathn_s^1 | \pathn_s^2) \trail_i \ldots \trail_{n+1}$, thus $\pathn^\prime_1 \pathn_c$ is an instance of $\trail$ if and only if $\pathn_n$ is.

In the case of ``greater than'' counting, we conclude immediately by induction as the same nodes are selected (thus there are enough). In the case of ``less than'', we need to check that no new node is counted in the smaller tree. Assume it is not the case for the formula $\countf{}{\trail}{\restrictedform}{\leq}{\natn}$, thus there is a path $\pathn_n \in \trail$ to a node satisfying $\restrictedform$. As the same node can be reached in $\stree$, and as we have $\stree(\pathn^\prime \pathn_n) \vdash^{\form} \neg \restrictedform$ by induction, we have a contradiction.

This concludes the proof when neither $\node_1$ nor $\node_2$ is annotated. The proof is identical when $\node_2$ is annotated. If $\node_1$ is annotated, we look at the first modality between $\node_1$ and $\node_2$. If it is a $\fc$, then we build the smaller tree by doing $\brel^\prime(\node_1,\fc) = \brel(n_2,\fc)$ (we remove the $\ns$ subtree from $\node_2$ instead of $\node_1$). Symmetrically, if the first modality is a $\ns$, we consider  $\brel^\prime(\node_1,\ns) = \brel(\node_2,\ns)$ as smaller tree. The rest of the proof proceeds as above.
\end{proof}

\begin{theorem}[Completeness]\label{thm:completeness}
  If $\form$ is satisfiable, then a satisfying tree is built.
\end{theorem}

\begin{proof}
  The proof proceeds as in \cite{geneves-pldi07}, we only need to check there are enough copies of each node to build every path. Let $\pathn$ be a path from the root of the tree to the leaves. By Lemma~\ref{lem:annotation_present}, there are at most $n+1$ identical nodes in this path, where $n$ is the number of annotations. The number of annotations is $c+1$ where $c$ is the number of counted nodes. We show by an immediate induction on the formula $\form$ that $c$ is bound by $K(\form)$ as defined in Figure~\ref{fig:boundK}. We conclude by remarking that $K(\form) + 2$ is the number of identical nodes we allow in the algorithm.
\end{proof}

\subsection{Complexity} \label{complexity}

We now show that the time complexity of the satisfiability algorithm is exponential in the formula size. This is achieved in two steps: we first show that the lean size is linear in the formula size,
then we show that the algorithm has a single exponential complexity with relation to the lean size.

\begin{lemma}
The lean size is linear in terms  of the original formula size.
\end{lemma}
\begin{proof}[Proof Sketch]
First note that the size of the lean is the number of elements it contains; the size of each element does not matter.

It was shown in \cite{geneves-pldi07} that the size of the lean generated by a non-counting formula is linear with respect to the formula size.

We now describe the case for counting formulas.
The lean consists of propositions and of modal subformulas, including the ones generated by the navigation of counting formulas (Figure~\ref{fig-nav}).
Moreover, each counting formula adds one fresh counting proposition. 
In the case of ``less than'' formulas $\countf{}{\trail}{\psi}{\leq}{k}$, a duplication occurs due to the consideration of the negated normal form of $\psi$. Since there is no counting under counting, this duplication and the fact that the negated normal form of a formula is linear in the size of the original formula (Figure~\ref{normalneg}) result in the lean remaining linear.
Another duplication  occurs in the case of counting formulas of the form $\countf{}{\trail_1|\trail_2}{\restrictedform}{\#}{k}$. This duplication does not double the size of the lean, however, since $\restrictedform$ still occurs only once in the lean, thus the number of elements in the lean induced by $\nav{\trail_1}{\restrictedform} \vee \nav{\trail_2}{\restrictedform}$ is the same as the sum of the ones in $\nav{\trail_1}{\restrictedform}$ and in $\nav{\trail_2}{\cdot}$.
\end{proof}

\begin{theorem}\label{complexitythm}

The satisfiability algorithm for the logic is decidable in time $2^{O(n)}$, where $n$ is the size of the lean.

\end{theorem}
\begin{proof}[Proof Sketch]
The maximum number of considered nodes is the number of distinct tree nodes which is $2^n$, the number of subsets of the lean. 
For a given formula $\phi$, the number of occurrences of the same node in the tree is bounded by $K(\form)\leq k*m$, where
$k$ is the greatest constant occurring in the counting formulas and $m$ is the number of counting subformulas of $\form$. 
Hence the number of steps of the algorithm is bounded by $2^n*k*m$.

At each iteration, the main operation performed by the algorithm is the composition of trees stored in $AUX$. The cost of each iteration consists in: the different searches needed to form the necessary triples \((\node,\stree_1,\stree_2)\), the  $\nmax$ function and $R^\form$. Since the total number of nodes is exponential, and the number of different subtrees too, therefore the maximum number of newly formed trees (triples) at each step has also an exponential bound. The function $\nmax$ performs a single traversal of the tree which is also exponential.
Since the entailment relation involved in the definition of $R^\form$ is local, $R^\form$ is performed in linear time. Computing the containment $AUX \subseteq ST$ and the union $ST \cup AUX$ are linear operations over sets of exponential size.

The stop condition of the algorithm is checked by the global entailment relation. It involves traversals parametrized by the number of trees, the number of nodes in each tree, the number of traversals for the entailment relation of counting formulas, and $K(\form)$. Its time complexity is bounded by $(2^n*k*m)^3$. 

Hence, the total time complexity of the algorithm is bounded by $(2^n*k*m)^{k^\prime}$, for some constant $k^\prime$.
\end{proof}

%%%%%%%%%%%%%%%%%%%%%%%%%%%%%%%%%%%%%%%%%%%%%%%%%%%%%%%%%%%%%%%%%%%%%%%%%

\section{Related Work}  \label{sec:relatedwork}

\paragraph{Counting over trees}
%Automata capable of expressing structural (schema) constraints and queries on finite trees were reported in  \cite{DBLP:conf/dbpl/CalvaneseGLV09}. Even more, it is able to perform backwards and recursive navigation, as well as to interpret nominals.

The notion of Presburger Automata for trees, combining both regular 
constraints on the children of nodes and numerical constraints given by Presburger 
formulas, has independently been introduced by Dal Zilio and Lugiez \cite{964013} and Seidl et al. \cite{DBLP:conf/icalp/SeidlSMH04}.  Specifically, Dal Zilio and Lugiez \cite{964013} propose a modal logic for unordered trees called Sheaves logic. This logic allows to impose certain arithmetical constraints on children nodes but lacks recursion (i.e., fixpoint operators) and inverse navigation.  Dal Zilio and Lugiez consider the satisfiability and the membership problems. 
%and they show  that Sheaves logic formulas can be translated into deterministic automata and is decidable.
 Demri and Lugiez \cite{DBLP:conf/cade/DemriL06} showed by means of an automata-free decision procedure that this logic is only  PSPACE-complete. Restrictions like {\em $\prop_1$ nodes have no more ``children'' than $\prop_2$ nodes}, are succinctly expressible by this approach. Seidl et al. \cite{DBLP:conf/icalp/SeidlSMH04} introduce a fixpoint Presburger logic, which, in addition to numerical constraints on  children nodes, also supports recursive forward navigation. For example, expressions like {\em the descendants of $\prop_1$ nodes have no more ``children'' than the number of children of descendants of $\prop_2$ nodes} are efficiently represented.
This means that constraints can be imposed on sibling nodes (even if they are deep in the tree) by forward recursive navigation but not on distant nodes which are not siblings. 

Compared to the work presented here, neither of the two previous approaches can efficiently support constraints like {\em there are more than $5$ ancestors of ``$\prop$'' nodes}. 

Furthermore, due to the lack of backward navigation, the works found in \cite{964013,DBLP:conf/icalp/SeidlSMH04,DBLP:conf/cade/DemriL06} are not suited for succinctly capturing  XPath expressions. Indeed, it is well-known that expressions with backward modalities are exponentially more succinct than their forward-only counterparts \cite{DBLP:conf/doceng/GenevesR05,685998}.

%For instance the path $\text{child::}a/\text{child::}b$ selects the nodes labeled by $b$ with an $a$ node as parent, which also has a parent. If we qualify such path $\text{self::}a/\text{child::}b[\text{child::}c/\text{child::}d]$, we filter such $b$ nodes to also have $c$ children with at least a $d$ child.
%Then, we focus on $b$ nodes, and there is a need to express a backwards navigation to express the path composition properties, as for the qualifying properties,
%they are expressed by downwards navigation.
%Although two-way navigation usually does not increase expressivity, it may come at an exponential cost \cite{685998}.

There is poor hope to push the decidability envelope much further for counting constraints.
Indeed, it is known from \cite{DBLP:conf/icalp/KlaedtkeR03,DBLP:conf/cade/DemriL06,Balder09} that the equivalence problem is undecidable  for XPath expressions with counting operators of the form:
\begin{itemize}
\item $\PathExpr_1[\qualifcount{\PathExpr_2} = \qualifcount{\PathExpr_3}]$,  or
\item $\PathExpr_1[\qualifposition = \qualifcount{\PathExpr_2}]$.
\end{itemize}
This is the reason why logical frameworks that allow comparisons between counting operators limit counting by restricting the $\PathExpr$  to immediate children nodes \cite{964013,DBLP:conf/icalp/SeidlSMH04}.
In this paper, we chose a different tradeoff: comparisons are restricted to constants but at the same time comparisons along more general paths are permitted. 

%Equipping modal logic with recursive and backward navigation together with Presburger constraints along recursive paths yields an undecidable calculus \cite{DBLP:conf/cade/DemriL06}.

%Results presented in this paper build on our preliminary work \cite{Barcenas09b,barcenas-planx09,geneves-pldi07}, and adopt a different tradeoff: we limit comparisons to constants, but we allow counting along more general paths. 

\paragraph{Counting over graphs}
The $\mu$-calculus is a propositional modal logic augmented with least and greatest fixpoint operators \cite{DBLP:conf/icalp/Kozen82}. Kupferman, Sattler and Vardi study a $\mu$-calculus with graded modalities where one can 
express, e.g., that a graph node has at least $n$ successors satisfying a certain property \cite{DBLP:conf/cade/KupfermanSV02}. The 
modalities are limited in scope since they only count immediate successors of a given node. A similar notion in trees consists in counting immediate children nodes, as performed by the counting formula $\modalf{\fc}{\countf{}{\ns^*}{\form}{\#}{k}}$, where $\phi$ describes the property to be counted. Compared to graded modalities of \cite{DBLP:conf/cade/KupfermanSV02}, we consider trees and we can extend the ``immediate successor'' notion to nodes reachable from regular paths, involving reverse and recursive navigation.

A recent study \cite{Bonatti-et-al-ICALP-06} focuses on extending the $\mu$-calculus with inverse modalities \cite{685998}, nominals \cite{DBLP:conf/cade/SattlerV01}, and graded modalities of \cite{DBLP:conf/cade/KupfermanSV02}. If only two of the above constructs are considered,  satisfiability of the enriched calculus is EXPTIME-complete \cite{Bonatti-et-al-ICALP-06,Bianco-lics09}.
However, if all of the above constructs are considered simultaneously, the calculus becomes undecidable \cite{Bonatti-et-al-ICALP-06}. 
The present work shows that this undecidability result in the case of graphs does not preclude decidable tree logics combining such features.

\paragraph{XPath-like counting extensions}\label{sec:gradedpaths}

% counting under fixpoints
% nested counting

The proposed logic can be the target for the compilation of a few more sophisticated counting features, considered as syntactic sugars (and that may come at the potential extra cost of their translation).

In particular, XPath allows nested counting, as  in the expression $$\text{self::book}[\text{chapter}[\text{section}>1]>1,$$ which selects the current ``book'' node provided it has at least two ``chapter'' child nodes which in turn must contain at least two ``section'' nodes each.
%Modeling this kind of expressions in general requires the extension of the logic to support nested counting formulas, like in: \[\text{book} \wedge \modalf{\fc}{\countf{}{\ns^*}{}{>}{1} \left(\text{chapter} \wedge \modalf{\fc}{\countf{}{\ns^*}{}{>}{1}\text{section} \right)}}\]
For a simple set of formulas, formulas that count only on children nodes, such nesting can be translated into ordinary logical formulas. 
%One possibility is to rely on a proper nesting of the property to be counted. 
For instance, the logical formulation of the above XPath expression can be captured as follows:
\[\text{book} \wedge \modalf{\fc}{\ufixpf{x}{\left(\text{chapter} \wedge \psi \wedge \modalf{\ns}{\ufixpf{y}{\text{chapter} \wedge \psi \vee \modalf{\ns}{y} }}\right) \vee \modalf{\ns}{x} }} 
\] where \(\psi=\modalf{\fc}{\ufixpf{x}{(\text{section} \wedge \modalf{\ns}{\ufixpf{y}{\text{section} \vee \modalf{\ns}{y} }}) \vee \modalf{\ns}{x} }} 
\).

In \cite{marx-tods05}, Marx introduced an  ``until'' operator for extending XPath's expressive power to be complete with respect to first-order logic over trees. This operator is trivially expressible in the present logic, owing to the use of the fixpoint binder.  We can even combine counting features with the ``until'' operator and express properties that go beyond the expressive power of the XPath standard. For instance, the following formula states that ``\emph{starting from the current node, and until we reach an ancestor named $a$, every ancestor has at least 3 children named $b$}'':
$$\ufixpf{x}{\left(\modalf{\fc}{\countf{}{\ns^*}{b}{>}{2}} \wedge \ufixpf{y}{\modalf{\invfc}{x}\vee \modalf{\invns}{y}}\right)  \vee a}$$

These extensions come at an extra cost, however. It is not difficult to observe (by induction) that, given a formula $\form$ with subformulas $\psi_1,...,\psi_n$ counting only on children nodes, if formulas $\psi_1,...,\psi_n$ are replaced by their expansions in $\form$, yielding a formula $\form'$, then $|lean(\form')|\leq |lean(\form)|*k^l$, where $k$ is greatest numerical constraint of the counting subformulas, and $l$ is the greatest level nesting of counting subformulas. 
As a consequence of Theorem~\ref{complexitythm}, the logic extended with nested formulas counting on children nodes and formulas counting on children nodes under the scope of a fixpoint operator can be decided in time $2^{O(n*k^l)}$. 
\section{Conclusion} \label{sec:conclusion}
We introduced a modal logic of trees equipped with (1) converse modalities, which allow to succinctly express forward and backward navigation, (2) a least fixpoint operator for recursion, and (3) cardinality constraint operators for expressing numerical occurrence constraints on tree nodes satisfying some regular properties. A sound and complete algorithm is presented for testing satisfiability of logical formulas.
This result is surprising since the corresponding logic for graphs is undecidable \cite{Bonatti-et-al-ICALP-06}. 

The decision procedure for the logic is exponential time w.r.t. to the formula size.
The logic captures regular tree languages with cardinality restrictions, as well as the navigational fragment of XPath equipped with counting features.  Similarly to backward modalities, numerical constraints do not extend the logical expressivity beyond regular tree languages. Nevertheless they enhance the succinctness of the formalism as they provide useful shorthands for otherwise exponentially large formulas.

This exponential gain in succinctness makes it possible to extend static analysis to a larger set of XPath and XML schema features in a more efficient way. We believe the field of application of this logic may go beyond the XML setting. For example, in verification of linked data structures \cite{DBLP:conf/lfcs/MannaSZ07,kuncak08,DBLP:journals/acta/HabermehlIV10} reasoning on tree structures with in-depth cardinality constraints seems a major issue. Our result may help building solvers that are attractive alternatives to those based on non-elementary logics such as SkS \cite{Thatcher68}, like Mona~\cite{mona_user_manual}.

%\paragraph{Future Work}
%We are currently implementing the satisfiability algorithm as an extension of our earlier solver introduced in \cite{geneves-pldi07,implementation}. 
%TODO: injective trails
%
%We say a trail is injective if its interpretation is an injection in any tree, that is, 
%if $\node\neq\node^\prime$, then %$\semf{\trail}{\tree}{\node}\cap\semf{\trail}{\tree}{\node^\prime}=\emptyset$.
%For example, the trail denoting the child/parent relation ($\fc,\ns^\star$) is injective: parents do not share children.

\bibliographystyle{plain}
\bibliography{ctl}

\end{document}